\documentclass[11pt]{article}
\usepackage{amsfonts,amsmath,amsthm,amssymb,dsfont, etex}

\usepackage{tikz}
\usetikzlibrary{arrows,automata}
\usepackage{comment, times}
\usepackage[margin=1in]{geometry}
\usepackage{soul,color}
\usepackage{graphicx,float,wrapfig}
\usepackage{mathrsfs}
\usepackage[usenames,dvipsnames]{pstricks}
\usepackage{epsfig}
\usepackage{pst-grad} 
\usepackage{pst-plot} 
\usepackage{mathdots}

\usepackage[linesnumbered,boxed,ruled,vlined]{algorithm2e}

\usepackage{url}

\usepackage{tabu}

\usepackage[normalem]{ulem}

\newtheorem{theo}{Theorem}[section]

\newtheorem{lemma}[theo]{Lemma}

\newtheorem{prop}[theo]{Proposition}
\newtheorem{cor}[theo]{Corollary}
\theoremstyle{definition}
\newtheorem{defi}[theo]{Definition}
\theoremstyle{plain}
\newtheorem{rem}[theo]{Remark}

\newenvironment{proofof}[1]{\begin{proof}[Proof of #1]}{\end{proof}}

\newenvironment{reminder}[1]{\bigskip
	\noindent {\bf Reminder of #1  }\em}{\smallskip}

\everymath{\displaystyle}

\newcommand{\poly}{\operatorname*{poly}}

\newcommand{\AMcc}{\textsf{AM}^{\cc}} 
\newcommand{\MA}{\textsf{MA}}

\newcommand{\PP}{\textsf{PP}}
\newcommand{\PPcc}{\textsf{PP}^{\cc}}

\newcommand{\NP}{\textsf{NP}}

\newcommand{\BP}{\textsf{BP}}
\newcommand{\UPP}{\textsf{UPP}}
\newcommand{\UPPcc}{\textsf{UPP}^{\cc}} 
\newcommand{\cc}{\textsf{cc}}
\newcommand{\eps}{\epsilon}

\newcommand{\polylog}{\operatorname*{polylog}}

\newcommand{\AND}{\textsf{AND}}
\newcommand{\OR}{\textsf{OR}}

\renewcommand{\epsilon}{\varepsilon}

\def\ShowAuthNotes{1}
\ifnum\ShowAuthNotes=1
\newcommand{\authnote}[2]{\ \\ \textcolor{red}{\parbox{0.9\linewidth}{[{\footnotesize {\bf #1:} { {#2}}}]}}\newline}
\else
\newcommand{\authnote}[2]{}
\fi

\newcommand{\lnote}[1]{\authnote{Lijie}{#1}}

\let\svfootnoterule\footnoterule
\renewcommand\footnoterule{\vfill\svfootnoterule}
\newcommand{\SAT}{\textsf{SAT}}
\newcommand{\SETH}{\textsf{SETH}}

\newcommand{\MaxIP}{\textsf{Max-IP}}
\newcommand{\IntMaxIP}{\textsf{$\mathbb{Z}$-Max-IP}}

\newcommand{\IntOV}{\textsf{$\mathbb{Z}$-OV}}
\newcommand{\Hopcroft}{\IntOV}

\newcommand{\logstar}{\log^{*}}
\newcommand{\WT}{\widetilde}

\newcommand{\BCPtwo}{\textsf{Bichrom.-$\ell_2$-Closest-Pair}}
\newcommand{\BCPone}{\textsf{Bichrom.-$\ell_1$-Closest-Pair}}

\newcommand{\FPtwo}{\textsf{$\ell_2$-Furthest-Pair}}

\newcommand{\MAXSAT}{\textsf{MAX-SAT}}
\newcommand{\ENP}{\textsf{E}^{\textsf{NP}}}
\newcommand{\SYM}{\textsf{SYM}}
\newcommand{\SYMAND}{\SYM\circ\AND}
\newcommand{\THR}{\textsf{THR}}
\newcommand{\THRTHR}{\THR\circ\THR}
\newcommand{\SYMTHR}{\SYM\circ\THR}
\newcommand{\SUM}{\textsf{SUM}}
\newcommand{\DISJOR}{\textsf{DOR}}
\newcommand{\GAPOR}{\textsf{Gap-OR}}
\newcommand{\MAJ}{\textsf{MAJ}}
\newcommand{\ETHR}{\textsf{ETHR}}
\newcommand{\NEXP}{\textsf{NEXP}}
\newcommand{\WMaxIP}{\textsf{Weighted-Max-IP}}

\newcommand{\LT}{\textsf{LT}}
\newcommand{\HLT}{\widehat{\textsf{LT}}}

\newcommand{\NC}{\textsf{NC}}

\newcommand{\APSP}{\textsf{APSP}}
\newcommand{\EMAJ}{\textsf{EMAJ}}

\newcommand{\KSAT}{\textsf{$k$-SAT}}
\newcommand{\TC}{\textsf{TC}}
\newcommand{\TIME}{\textsf{TIME}}

\newcommand{\SUMOR}{\textsf{SUMOR}}

\newcommand{\psirevx}{\psi^{x}_\textsf{rev}}
\newcommand{\psirevy}{\psi^{y}_\textsf{rev}}

\newcommand{\MAX}{\textsf{Max}}

\newcommand{\AMA}{\textsf{AMA}}

\newcommand{\ckt}{\mathscr{C}}
\newcommand{\com}{\textsf{com}}
\newcommand{\IND}{\textsf{IND}}
\newcommand{\dt}{\textsf{dt}}

\newcommand{\ACC}{\textsf{ACC}}
\newcommand{\NQP}{\textsf{NQP}}

\newcommand{\MAXkSAT}[1][k]{\textsf{MAX-$#1$-SAT}}

\title{Toward Super-Polynomial Size Lower Bounds for Depth-Two Threshold Circuits}

\author{Lijie Chen\thanks{Email: lijieche@mit.edu. Supported by an Akamai Fellowship.}\\MIT}

\date{}

\begin{document}
	\clearpage\maketitle
	\thispagestyle{empty}
	\begin{abstract}
		Proving super-polynomial size lower bounds for $\TC^0$, the class of constant-depth, polynomial-size circuits of Majority gates, is a notorious open problem in complexity theory. A major frontier is to prove that $\NEXP$ does not have poly-size $\THRTHR$ circuit (depth-two circuits with linear threshold gates). 
		
		In recent years, R.~Williams proposed a program to prove circuit lower bounds via improved algorithms. In this paper, following Williams' framework, we show that the above frontier question can be resolved by devising \emph{slightly} faster algorithms for several fundamental problems:        
		
		\begin{itemize}
			\item \textbf{Shaving Logs for \FPtwo.} An $n^2 \poly(d) / \log^{\omega(1)} n$ time algorithm for \FPtwo\ in $\mathbb{R}^d$ for polylogarithmic $d$ implies $\NEXP$ has no polynomial size $\THRTHR$ circuits. The same holds for Hopcroft's problem, $\BCPtwo$ and Integer $\MaxIP$.
			
			\item \textbf{Shaving Logs for Approximate $\BCPtwo$.} An $n^2 \poly(d) / \log^{\omega(1)} n$ time algorithm for $(1+1/\log^{\omega(1)} n)$-approximation to $\BCPtwo$ or \BCPone\ for polylogarithmic $d$ implies $\NEXP$ has no polynomial size $\SYMTHR$ circuits.
			
			\item \textbf{Shaving Logs for Modest Dimension Boolean $\MaxIP$.} An $n^2 / \log^{\omega(1)} n$ time algorithm for Bichromatic Maximum Inner Product with vector dimension $d = n^\eps$ for any small constant $\eps$ would imply $\NEXP$ has no polynomial size $\THRTHR$ circuits. Note there is an $n^2\polylog(n)$ time algorithm via fast rectangle matrix multiplication.
		\end{itemize}
	
		Our results build on two structure lemmas for threshold circuits: a poly-size $\THRTHR$ circuit can be written as
		
		\begin{itemize}
			\item an $\OR$ of polynomially many poly-size $\THR\circ\MAJ$ circuits;
			
			\item an $\OR$ of sub-exponentially many poly-size $\MAJ\circ\MAJ$ circuits, or as an $\OR$ of polynomially many sub-exponential size $\MAJ\circ\MAJ$ circuits.
		\end{itemize}
		The second structure lemma itself only gives a randomized reduction, which we derandomize nondeterministically to apply Williams' connection. 
		
		With similar techniques, we also show slightly improved algorithms for $\MAXSAT$ and $\KSAT$ would imply interesting circuit lower bounds:
		
		\begin{itemize}
			\item \textbf{Better Algorithms for $\MAXSAT$ Implies Super-quasi-polynomial $\SYMAND$ Lower Bounds.} A $2^{n \cdot (1 - 1/2^{ (\log m)^{o(1)} } ) }$ time algorithm for $\MAXSAT$ implies that $\NEXP$ has no quasi-polynomial size $\SYMAND$ circuits. This is to be contrasted with $\textsf{CNF-}\SAT$, which admits a $2^{n \cdot (1 - 1/\log(m/n))}$ time algorithm.
			\item \textbf{Better Algorithms for $\KSAT$ Breaks the $\log\log n$ Depth Barrier for $\TC$ Circuits.} An algorithm for $\KSAT$ in $2^{n \cdot (1 - 1/k^{1 / \omega(\log\log k)} )}$ time implies that $\ENP$ has no linear size (in terms of wires) $O(\log\log n)$-depth $\TC$ circuits. The best known algorithm runs in $2^{n \cdot (1 - 1 / O(k))}$ time.            
		\end{itemize}
		
	\end{abstract}
	\addtocounter{page}{-1}
	\newpage

\section{Introduction}

What interesting functions do not have polynomial-size $\TC^0$ circuits? Despite substantial research effort on this question~\cite{HajnalMPST93,AllenderK10,AmanoM05,ChenS15,ForsterKLMSS01,GoldmannHR92,groeger1993linear,HansenP10,HansenP15,ImpagliazzoPS13,ImpagliazzoPS97,nisan1993communication,PaturiS94,RoychowdhuryOS94,Williams14THR,Tamaki16,AlmanCW16,KaneW16} it is consistent with current knowledge that $\NEXP$ has polynomial-size $\THRTHR$ or $\SYMTHR$ circuits\footnote{$\THRTHR$ refers to depth-$2$ circuits consisting of linear threshold gates. $\SYMTHR$ refers to depth-$2$ circuits consisting of a top $\SYM$ gate and many bottom $\THR$ gates. See Section~\ref{sec:circuits} for formal definitions.}.

In 2011, a breakthrough result of R.~Williams~\cite{Wil14ACC,Wil13} showed that $\NEXP$ does not have polynomial-size $\ACC^0$ circuits, by connecting an appealing algorithmic approach to circuit lower bounds: circuit lower bounds can be proved by slightly-better-than-trivial circuit-analysis algorithms for problems such as satisfiability or derandomization. Along these lines, several subsequent works follow Williams' program~\cite{Williams13NEXPBPP,Williams14THR,Ben-SassonV14,JahanjouMV15,AlmanCW16,Williams16Derand,Tamaki16}, and lower bounds for more circuit classes have been proved by introducing new algorithms, or tightening the connection itself. For an example of the latter, in the recent exciting work by Murray and Williams~\cite{MurrayW17}, it is shown that $\sf{NTIME}[n^{\poly(\log n)}]$ does not have polynomial-size $\ACC^0 \circ \THR$ circuits, via a new Easy Witness Lemma.

The next big challenge for complexity theorists would be to apply Williams' connection to prove that $\NEXP$ (even $\NQP$\footnote{$\sf{NTIME}[n^{\poly(\log n)}]$}) is not contained in depth-$2$ threshold circuits. In fact, partial results are already made. In~\cite{Tamaki16,AlmanCW16}, it is shown that $\ENP$ is not contained in $n^{2-o(1)}$ size $\THRTHR$ circuits.

In this paper, we apply Williams' connection, together with many new and old tools from the structure theory of threshold circuits, to show that super-polynomial circuit lower bounds for $\THRTHR$ or $\SYMTHR$ would follow from tiny improvements (shaving all polylogs) over the running time of many fundamental problems in \emph{computational geometry}.

We also consider two other well-studied fundamental problems $\MAXSAT$ and $\KSAT$: the canonical $\textsf{NP-hard}$ optimization problem and the canonical $\textsf{NP-complete}$ problem. The state-of-the-art algorithms for $\MAXSAT$ are much slower than that of $\textsf{CNF-SAT}$, and the best known running time for $\KSAT$ has remained at $2^{n (1 - 1/O(k))}$ for 20 years. We show that (very) modest improvements on their current state-of-the-art algorithms would imply lower bounds for $\SYMAND$ circuits, and for $O(\log\log n)$-depth $\TC$ circuits. These results for $\MAXSAT$ and $\KSAT$ can be interpreted in two ways: either as a barrier for getting faster algorithms because proving circuit lower bounds is generally considered hard, or as a new approach for attacking those long-standing open questions in circuit complexity, providing extra motivations for studying these two problems.

\subsection{Our Results}

\subsection*{Consequence of Shaving Logs from $\FPtwo$ and Related Problems}

Our first result is that shaving logs from $\FPtwo$ or other related problems in computational geometry would resolve our open problem in circuit complexity.

\begin{theo}\label{theo:Hopcroft-THRTHR-NEXP}
	If any of the following problems has an $n^{2} \poly(d) / \log^{\omega(1)} n$ time deterministic algorithm for polylogarithmic $d$, then $\NEXP$ has no polynomial size $\THR \circ \THR$ circuits:
	
	\begin{enumerate}
		\item $\Hopcroft_{n,d}$ (Hopcroft's Problem): Find an orthogonal pair among $n$ points in $\mathbb{Z}^d$.
		
		\item $\FPtwo_{n,d}$:  Find the $\ell_2$-furthest pair among $n$ points in $\mathbb{R}^d$.
		
		\item $\BCPtwo_{n,d}$: Given two set $A,B$ of $n$ points in $\mathbb{R}^{d}$, compute $\min_{(a,b) \in A \times B} \|a - b\|_2$.
		
		\item $\IntMaxIP_{n,d}$: Given two sets $A,B$ of $n$ vectors from $\mathbb{Z}^{d}$, compute $\max_{(a,b) \in A \times B} a \cdot b$.
		
		\item $\WMaxIP_{n,d}$: Given a weight vector $w \in \mathbb{Z}^{d}$ and two sets $A,B$ of $n$ vectors from $\{0,1\}^{d}$, compute $\max_{(a,b) \in A \times B} a \odot_{w} b$, where $a \odot_{w} b := \sum_{i=1}^{d} w_i \cdot a_i \cdot b_i $.
	\end{enumerate}
\end{theo} 

The best known algorithms for $\Hopcroft$, $\FPtwo$, $\BCPtwo$ and $\IntMaxIP$ are of running time $n^{2 - 1/O(d)}$~\cite{matouvsek1992efficient,agarwal1991euclidean,yao1982constructing}, which means there is no improvement when $d = \Omega(\log n)$. But note that we do not require a truly-subquadratic time algorithm here: we only need to ``shave all the logs'' from the trivial $n^2 \poly(d)$ running time for polylogarithmic $d$, and we only need to do so for one of the above problems. We are optimistic that such algorithms exist, given the rich toolkit (which keeps growing) available for solving geometry problems.

We also remark here that all problems above except for the last one requires $n^{2 - o(1)}$ time when $d = 2^{O(\logstar n)}$ under $\SETH$~\cite{Che18}. But again that conditional lower bound says nothing about whether shaving logs are possible.

\subsection*{Consequence of Shaving Logs for Approximate \BCPtwo}

Our second result is that shaving logs on problems which are easier than those in the previous section would imply circuit lower bounds for $\SYMTHR$.

\begin{theo}\label{theo:Max-IP-NEXP}
	If any of following problems has an $n^{2} \poly(d) / \log^{\omega(1)} n$ time deterministic algorithm for polylogarithmic $d$, then $\NEXP$ has no polynomial size $\SYM \circ \THR$ circuits:
	
	\begin{enumerate}
		\item $\MaxIP_{n,d}$: Given two sets $A,B$ of $n$ vectors from $\{0,1\}^{d}$, compute $\max_{(a,b) \in A \times B} a \cdot b$.
		
		\item Compute a $(1 + 1/\log^{\omega(1)} n)$-approximation to $\BCPtwo_{n}$.
		
		\item Compute a $(1 + 1/\log^{\omega(1)} n)$-approximation to $\BCPone_{n}$.
	\end{enumerate}
\end{theo}

The best known algorithms for $(1+\eps)$-approximation to $\BCPone$ or \BCPtwo\ runs in $n^{2 - \widetilde{\Omega}(\eps^{1/3})}$ time, while the best known algorithm for $\MaxIP_{n,d}$ runs in $n^{2 - \widetilde{\Omega}(1 / \sqrt{d/\log n})}$ time (for $d \gg \log n$). For those algorithms, there is no improvement when $\eps \ll 1/\log^3 n$ or $d \gg \log^3 n$.  Also, note that all these problems require $n^{2 - o(1)}$ time when $d = \omega(\log n)$ under $\SETH$~\cite{Rubinstein2017closest,Williams05}. But again it seems plausible that there are some clever ways to shave logs in higher dimensional cases. 

A more fine-grained statement can be made if we relax $\NEXP$ to $\ENP$:

\begin{theo}\label{theo:Max-IP-ENP}
	Suppose for a real $k > 2$, one of the following deterministic algorithms exists:
	
	\begin{enumerate}
		\item An $n^2 / \log^{\omega(1)} n$ time algorithm for $\MaxIP_{n, \log^k n}$.
				
		\item A $(1 + 1/\log^k n)$-approximation algorithm for $\BCPone_{n}$ in $n^2 / \log^{\omega(1)} n$ time.
		
		\item A $(1 + 1/\log^k n)$-approximation algorithm for $\BCPtwo_{n}$ in $n^2 / \log^{\omega(1)} n$ time.
	\end{enumerate}
	
	Then $\ENP$ has no $n^{(k-2)/2 - \eps_1}$-size $\SYM \circ \SYM$ circuits for any $\eps_1 > 0$. 
\end{theo}
\begin{rem} Note that 
	$\SYM\circ\SYM$ and $\SYMTHR$ are equivalent up to a polynomial size blow-up~\cite{HansenP10,GoldmannHR92}. See also Proposition~\ref{prop:circuit-facts-contain} (4).
\end{rem}

\subsection*{Consequence of Shaving Logs for Modest Dimension Boolean $\MaxIP$}

Our third result is that shaving logs from moderate dimension $\MaxIP$ would imply super-polynomial lower bound for $\THRTHR$.

\begin{theo}\label{theo:MaxIP-THRTHR-NEXP}
	If any of the following deterministic algorithms exists, then $\NEXP$ has no polynomial-size $\THRTHR$ circuits:
	
	\begin{enumerate}
		\item An algorithm solving $\MaxIP_{n,n^\eps}$ in $n^{2} / \log^{\omega(1)}(n)$ time, for a constant $\eps > 0$.
		
		\item An algorithm solving $\MaxIP_{n,\log^k(n)}$ in $n^{2-\eps}$ time for a constant $\eps > 0$ and any integer $k$.
	\end{enumerate}
\end{theo}

Note that for small enough $\eps > 0$, $\MaxIP_{n,n^\eps}$ can be solved in $n^2 \polylog(n)$ time by applying the fast rectangle matrix multiplication algorithm~\cite{coppersmith1982rapid} to calculate the pair-wise inner products. Therefore, we only need to \emph{shave logs} on this naive algorithm.

\subsection*{Two Structure Lemmas for $\THRTHR$ circuits}

The major technical ingredients of our results are two structure lemmas for $\THRTHR$, of interest in its own right.

Informally, the first lemma says every $\THRTHR$ is equivalent to a polynomial OR of Threshold-of-Majority circuits and the second lemma says that every $\THRTHR$ circuit is equivalent to a ``subexponential OR'' of Majority-of-Majority circuits. For the program of proving $\THRTHR$ lower bounds, this is significant, as exponential-size Majority-of-Majority and Threshold-of-Majority lower bounds are well-known~\cite{HajnalMPST93,ForsterKLMSS01}.

In the following, $\DISJOR$ refers to a ``disjoint'' $\OR$ gate: an $\OR$ gate with the promise that at most one of its inputs is ever true, and $\GAPOR$ refers to a ``gapped'' $\OR$ gate: an $\OR$ gate with the promise that either all inputs are false or at least half of the inputs are true. (See Section~\ref{sec:circuits} for formal definitions.)

\begin{lemma}[Structure Lemma I for $\THRTHR$ circuit]\label{lm:structure-THRTHR-I}
	Let $n$ be number of inputs and $s = s(n) \ge n$ be a size parameter.
	Every $s$-size $\THRTHR$ circuit $C$ is equivalent to a $\GAPOR \circ \THR \circ \MAJ$ circuit
	such that:
	
	\begin{itemize}
		\item The top $\GAPOR$ gate has $\poly(s)$ fan-in.
		\item Each sub $\THR \circ \MAJ$ circuit has size $\poly(s)$.
	\end{itemize}
	
	Moreover, the reduction can be computed in \emph{deterministic} $\poly(s)$ time.
\end{lemma}

\begin{lemma}[Structure Lemma II for $\THRTHR$ circuit]\label{lm:structure-THRTHR-II}
	Let $n$ be number of inputs and $s = s(n)$ be a size parameter. Let $\eps \in \left(\frac{\log s}{n},1 \right)$. For $s = 2^{o(n)}$, every $s$-size $\THRTHR$ circuit $C$ is equivalent to a $\DISJOR \circ \MAJ \circ \MAJ$ circuit  such that:
	
	\begin{itemize}
		\item The top $\DISJOR$ gate has $2^{ O(\eps n)} \cdot \poly(s)$ fan-in.
		\item Each sub $\MAJ \circ \MAJ$ circuit has size $s^{O(1/\eps)} \cdot \poly(s)$.
	\end{itemize}
	
	Moreover, the reduction can be computed in \emph{randomized} $2^{O(\eps n)} \cdot s^{O(1/\eps)} \cdot \poly(s)$ time. 
\end{lemma}

We discuss some immediate applications of the structure lemmas.

\paragraph*{Equivalence of Non-trivial $\SAT$ Algorithms.} It is well-known that $\THRTHR$ circuits can be simulated with depth-$3$ polynomial $\MAJ\circ\MAJ\circ\MAJ$ circuits~\cite{GoldmannHR92}; however, replacing the output $\MAJ$ gate with an extremely simple $\GAPOR$ or $\DISJOR$ gate has extra benefits. For example, any faster SAT algorithm for $\THR \circ \MAJ$ or $\MAJ \circ \MAJ$ circuits can be used to obtain a SAT algorithm for $\THRTHR$ easily! Formally, the following two corollaries follow from Lemma~\ref{lm:structure-THRTHR-I} and Lemma~\ref{lm:structure-THRTHR-II} directly. 

\begin{cor}\label{cor:equivalent-THRTHR-THRMAJ}
	The following are equivalent:
	\begin{itemize}
		\item The satisfiability of $\THRTHR$ circuits of size $n^k$ can be solved in $2^{n} / n^k$ time for any $k$.
		\item The satisfiability of $\THR\circ\MAJ$ circuits of size $n^k$ can be solved in $2^{n} / n^k$ time for any $k$.
	\end{itemize}
\end{cor}

\begin{cor}\label{cor:equivalent-THRTHR-MAJMAJ}
	The following are equivalent:
	
	\begin{itemize}
		\item There is a $2^{(1-\Omega(1)) \cdot n}$ time algorithm for the satisfiability of polynomial size $\THRTHR$ circuits.
		\item There is a $2^{(1-\Omega(1)) \cdot n}$ time algorithm for the satisfiability of polynomial size $\MAJ\circ\MAJ$ circuits.
	\end{itemize}
\end{cor}

We remark that the first corollary preserves any non-trivial speed up ($2^n / n^{\omega(1)}$ time algorithms), while the second is coarser, which is due to the sub-exponential blowup (when $\eps$ is an arbitrarily small constant) in Lemma~\ref{lm:structure-THRTHR-II}. 

\paragraph*{Generalization to Threshold Circuits of Constant Depth.} Also, Lemma~\ref{lm:structure-THRTHR-II} easily generalizes to threshold circuits of any constant depth. In the following $\LT_d$ denotes threshold circuits of depth-$d$, while $\HLT_d$ denotes depth-$d$ majority circuits (see Section~\ref{sec:circuits} for formal definitions).

\begin{cor}\label{cor:structure-LT}
	Let $n$ be number of inputs and $s = s(n)$ be a size parameter. Let $\eps \in \left(\frac{\log s}{n},1 \right)$ and $d$ be a constant. For $s = 2^{o(n)}$, every $s$-size $\LT_d$ circuit is equivalent to a $\DISJOR \circ \HLT_d$ circuit such that:
	\begin{itemize}
		\item The top $\DISJOR$ gate has $2^{ O(\eps \cdot n) }$ fan-in.
		\item Each sub $\HLT_d$ circuit has size $O\left( s^{O(1/\eps)} \right)$.
	\end{itemize}
\end{cor}

\paragraph*{Structure Lemma for Polynomial Threshold Functions.} Our ideas can also be used to derive a structure lemma for polynomial threshold functions of degree $k$, i.e., $\THR \circ \AND_{k}$ circuits:

\begin{cor}\label{cor:structure-THR-AND}
	Let $n$ be number of inputs and $s = s(n)$ be a size parameter. Let $\eps \in \left(\frac{\log s}{n},1 \right)$ and $k$ be a constant. Assuming $s = 2^{o(n)}$, an $s$-size $\THR \circ \AND_k$ circuit is equivalent to a $\DISJOR \circ \MAJ \circ \AND_{2k}$ circuit such that:
	
	\begin{itemize}
		\item The top $\DISJOR$ gate has $2^{ O(\eps \cdot n) }$ fan-in.
		\item Each sub $\MAJ \circ \AND_{2k}$ circuit has size $O\left( s^{O(1/\eps)} \right)$.
	\end{itemize}
	
	The above still holds if we replaced both $\AND_{k}$ and $\AND_{2k}$ by unbounded fan-in $\AND$ gates.
\end{cor}

That is, every polynomial threshold function of degree $k$ with arbitrary weights can be simulated by a \emph{subexponential-size} disjoint OR of polynomial threshold functions of degree $2k$ with small weights.

The following corollary follows from that the SAT problem for $\THR \circ \AND_{k}$ circuits is equivalent to the \emph{weighted} $\MAXkSAT[k]$ problem (given a $\textsf{CNF}$ formula $\varphi$ with weights on each clause, find an assignment satisfying clauses of maximum total weight), and that SAT for $\MAJ \circ \AND_{2k}$ is equivalent to the (unweighted) $\MAXkSAT[2k]$ problem. 

\begin{cor}\label{cor:weight-unweight-MAXSAT-equiv}
	For any integer $k$, if there is a $2^{(1-\Omega(1)) n}$ time algorithm for polynomial size unweighted $\MAXkSAT[2k]$, then so does polynomial size weighted $\MAXkSAT$.\footnote{We assume the weights are at most $2^{\poly(n)}$ for making the input polynomial size.}
\end{cor}

\newcommand{\IP}{\textsf{IP}}
\newcommand{\RP}{\textsf{RP}}
\newcommand{\UP}{\textsf{UP}}

\paragraph*{An Application in Communication Complexity.} Finally, Structure Lemma I also has an application in communication complexity.

The connection between threshold circuits and communication lower bounds dates back to \cite{nisan1993communication}, which showed $\MAJ \circ \THR$ circuits admit efficient $\PP^\cc$ protocols, therefore the $\Omega(n)$ $\PPcc$ lower bound for IP2 (Inner Product)~\cite{CG85PPLowb} implies that IP2 requires $2^{\Omega(n)}$-size $\MAJ \circ \THR$ circuits. Later, \cite{ForsterKLMSS01} showed that $\THR \circ \MAJ$ circuits have efficient $\UPPcc$ protocols, hence the $2^{\Omega(n)}$ $\THR\circ\MAJ$ circuit size lower bound for IP2 can be derived from the corresponding $\UPP^\cc$ lower bound for it~\cite{Forster02Signrank}. 

Naturally, one may seek similar connections for $\THR \circ \THR$ circuits. In a recent work,~\cite{AW17Rigidity} showed that $\THR \circ \THR$ admits efficient $\BP \cdot \UPP^{\cc}$ protocols. However, it seems quite hard to prove $\BP \cdot \UPPcc$ lower bounds, as it contains $\AMcc$, which itself has been notoriously hard to prove any non-trivial lower bound on~\cite{GPW18CCLandscape}.

In this work, building on our Structure Lemma I for $\THRTHR$ circuits. We show that $\THRTHR$ circuits admit efficient $\RP \cdot \UPPcc$ protocols, which could be potentially easier to prove a lower bound on.

\begin{theo}\label{theo:new-protocols-THRTHR}
	For a function $F : \{0,1\}^{n} \times \{0,1\}^n \to \{0,1\}$, suppose it admits a $\THRTHR$ circuit of size $s$, then it also admits a $\RP \cdot \UPP^\cc$ protocol of cost $O(\log s)$.
\end{theo}
\begin{rem}
	Therefore, a $\log^{\omega(1)} n$ lower bound on $\RP \cdot \UPPcc$ complexity for a function $F$, implies that $F$ requires $n^{\omega(1)}$-size $\THRTHR$ circuits.
\end{rem}

See Appendix~\ref{app:cc} for details and a formal definition of $\RP \cdot \UPPcc$ protocols.

\subsection*{$\MAXSAT$ and $\SYMAND$ Circuit Lower Bounds}

Being the canonical $\NP$-hard optimization problem, a huge amount of research effort has been devoted to finding faster-than-$2^n$ algorithms for $\MAXSAT$~\cite{BansalR99,Binkele-RaibleF10,BliznetsG12,ChenK04,DantsinW06,GaspersS12,GaspersS17,GolovnevK14,GrammHNR03,GrammN00,Hirsch00,Hirsch03,KneisMRR05,KojevnikovK06,Kulikov05,KulikovK07,MahajanR99,NiedermeierR00,SakaiST15,ScottS03,Williams05}. 

In turns of getting non-trivial speed up, in~\cite{SakaiSTT16}, a $2^{n - n^{1/O(t)}}$ time algorithm for $\MAXSAT$ of $m = n^t$ clauses is proposed, which doesn't give any improvement when $m = n^{\Omega(\log n)}$. 

This state of affairs is certainly unsatisfactory, as $\KSAT$ and $\textsf{CNF-SAT}$ are known to have much better algorithms: $\KSAT$ is solvable in $2^{(1 - 1/O(k))n}$ time~\cite{PaturiPSZ05}, while $\textsf{CNF-SAT}$ admits a $2^{(1 - \Omega(1/\log(m/n))n}$ time algorithm~\cite{CalabroIP06,DantsinH09}, which gives a non-trivial speedup even for sub-exponential $m$. 

Our next result give some evidence why progress on $\MAXSAT$ has been limited. We show that a very modest improvement over the best known $\MAXSAT$ algorithm would imply super-quasi-polynomial circuit lower bounds for $\SYMAND$ circuits.

\begin{theo}\label{theo:MAX-SAT-NEXP}
	If there is an algorithm for $\MAXSAT$ solving an instance with $2^{\log^k n}$ clauses in $2^{n} / 2^{\log^k n}$ time for every integer $k$. Then $\NEXP$ has no quasi-polynomial size $\SYMAND$ circuits.
\end{theo}
\begin{cor}
	An algorithm for $\MAXSAT$ with $m$ clauses in 
	\[
	2^{n \cdot \left(1 - 1/2^{ (\log m)^{o(1)} } \right)} \text{ time}
	\] 
	implies that $\NEXP$ has no quasi-polynomial size $\SYMAND$ circuits.
\end{cor}

Moreover, if the running time for $\MAXSAT$ can be improved to the same as the best-known algorithms for $\textsf{CNT-SAT}$, then we would have a much stronger circuit lower bound.

\begin{theo}\label{theo:MAX-SAT-ENP}
	If there is a $2^{n - \Omega(n/\log m)}$ time algorithm for $\MAXSAT$ with $m$ clauses, then $\ENP$ has no $2^{o(\sqrt{n})}$-size $\SYMAND$ circuit.
\end{theo}

\paragraph*{$\SYMAND$ Circuits Lower Bounds.}
The best non-trivial lower bound for $\SYMAND$ is an $n^{\Omega(\log n)}$ size lower bound for a function in $\ACC^0$~\cite{RazborovW93}. Historically, quasi-polynomial size $\SYMAND$ circuits are studied mainly because it contains the circuit class $\ACC^0$ by depth-reduction~\cite{Yao90,BeigelT94,AllenderG94}, and it is connected to Number-On-Forehead communication protocols~\cite{HastadG91}. So it was considered as a viable approach to prove $\ACC^0$ circuit lower bounds. However, even Williams' breakthrough work on $\ACC^0$ makes crucial use of the depth-reduction result $\ACC^0 \subseteq \SYMAND$, it is not clear how to prove super quasi-polynomial lower bound on $\SYMAND$ itself via his approach, which is unsatisfying.

This question is also interesting as exponential lower bounds for $\MAJ \circ \AND$ are trivial to obtain (in fact, even exponential lower bound for $\THR \circ \MAJ$ are known~\cite{ForsterKLMSS01}). It would be good to understand why switching the top gate to a $\SYM$ gate makes the problem much harder.

\subsection*{$\KSAT$ and the $\log \log n$-Depth Barrier for $\TC$ Circuits}

Finally, we show a modestly improved algorithm for $\KSAT$ (recall the state-of-the-art running time is $2^{n \cdot (1 - 1/O(k))}$) would imply lower bounds for $O(\log\log n)$-depth $\TC$ circuits. This is based on the reduction from $\textsf{TC-SAT}$ to $\KSAT$ in a recent work by Abboud et al.~\cite{abboud2017more}.

\begin{theo}\label{lm:KSAT-to-TC}
	A $2^{n \cdot (1 - k^{1 / \omega(\log \log k)} )}$ time algorithm for $\KSAT$ implies that for any constant $c > 1$, $\ENP$ has no $cn$-wire depth-$(c\log\log n)$ $\TC$ circuit.	
\end{theo}

\paragraph*{$\log\log n$-depth Barrier for $\TC$.}

In~\cite{ImpagliazzoPS97}, it is shown that parity requires $n^{1 + c^{-d}}$-wires for depth-$d$ $\TC$ circuits, which becomes linear when $d = \Omega(\log\log n)$. No non-trivial super-linear wires lower bounds are known when the depth is $\Omega(\log\log n)$. It is consistent with the current state of knowledge that $\ENP$ could be contained in linear-size $O(\log\log n)$-depth $\TC$ circuits.

\subsection{Related Works}

\paragraph*{Constant-Depth Threshold Circuit Lower Bounds.} For more history on previous works on lower bounds for constant-depth threshold circuits, see the corresponding sections in~\cite{Williams14THR,KaneW16}. We only discuss a few recent results here.

In 2014, Williams~\cite{Williams14THR} showed that $\NEXP$ is not contained in $\ACC_0 \circ \THR$, by devising a fast satisfiability algorithm for it. The lower bound was recently improved by Murray and Williams~\cite{MurrayW17} to that $\NQP$ is not contained in $\ACC_0 \circ \THR$. Tamaki~\cite{Tamaki16}, Alman, Chan and Williams~\cite{AlmanCW16} proved that $\ENP$ is not contained in $n^{2 - o(1)}$ size $\THRTHR$ circuits (the results in~\cite{AlmanCW16} is stronger, it in fact showed lower bound against $\ACC_0 \circ \THRTHR$ circuits, with at most $n^{2-\eps}$ bottom $\THR$ gates). Most recently, Williams~\cite{WilliamsLinearComb} showed that there are functions in $\NQP$ can not be represented by a linear combination of polynomially many $\ACC\circ\THR$ circuits.

Tell~\cite{Tell17} constructed a \emph{quantified derandomization} algorithm for $\TC$ circuits with depth $d$ and $n^{1 + \exp(-d)}$ wires, and showed that a modest improvement of his algorithm would imply standard derandomization of $\TC^0$, and consequently $\NEXP \not\subseteq \TC^0$.

Using random restriction, Kane and Williams~\cite{KaneW16} proved that any $\THRTHR$ circuits computing Andreev's function requires $\widetilde{\Omega}(n^{1.5})$ gates and $\widetilde{\Omega}(n^{2.5})$ wires. Chattopadhyay and Mande~\cite{ChattopadhyayM17a} showed an exponential size separation between $\THR\circ\MAJ$ and $\THR\circ\THR$, by constructing a function in $\THR\circ\THR$ with exponential \emph{sign-rank}.

\paragraph*{Shaving Logs Implies Circuit Lower Bound.} Abboud et al.~\cite{AHVW16} showed that shaving logs on some well-studied sequence alignment problems like Edit-Distance and Longest Common Subsequence would imply strong circuit lower bounds. In particular, they proved that an $n^2 / \log^{\omega(1)} n$ time algorithm for either of them would imply a $2^{n} / n^{\omega(1)}$ time algorithm for poly-size $\textsf{Formula-SAT}$, from which it follows $\NEXP$ is not contained in $\NC^1$. 
Their reduction is later tightened by Abboud and Bringmann~\cite{abboud2018tighter}, which showed that an $n^2 / \log^{7+\eps} n$ time algorithm for either of them is already enough to imply new algorithm for $\textsf{Formula-SAT}$. Chen et al.~\cite{Chen2018FineGrained} showed that shaving a $2^{(\log\log N)^3}$ factor from the naive algorithm for (constant factor) approximate $\textsf{Closest-LCS-Pair}$\footnote{Given two sets $A,B$ of strings, compute $\max_{(a,b) \in A \times B} \textsf{LCS}(a,b)$.} and many other problems also implies $\NEXP$ is not contained in $\NC^1$.

\section{Preliminaries}


We begin with some notations for operations on vectors. For two vectors $u, v \in \{0,1\}^*$, we use $u \circ v$ to denote their concatenation. For a vector $u$ and an integer $t$, we use $u^{\otimes t}$ to denote the vector obtained by repeating $u$ $t$ times.

\subsection{Circuits Classes}\label{sec:circuits}

Since we discuss many circuits class in this work, we begin with some notations for those classes.

\paragraph*{Notations for Circuit Classes.} Let $x \in \{0,1\}^{n}$ be a Boolean input. For $w \in \mathbb{R}^{n}$ and $t \in \mathbb{R}$, we define $\THR_{w,t}(x)$ (the threshold function) be the indicator function that whether $w \cdot x \ge t$. Similarly, we define $\ETHR_{w,t}(x)$ (the exact threshold function) be the indicator function that whether $w \cdot x = t$. The vector $w$ and the real $t$ are called the weights and the threshold of the given function $\THR_{w,t}$ ($\ETHR_{w,t}$). We say these weights and thresholds are \emph{realizations} of the Boolean functions they defined. Note that a function may have different realizations. One may assume without loss of generality that the weights and thresholds are integers of absolute value at most $2^{O(n \log n)}$~\cite{muroga1961theory,babai2010weights}. For a threshold or exact threshold function with weight $w$, we call the linear function $L := w \cdot x$ its associated linear function.


We use $\MAJ_n$ and $\EMAJ_n$ to denote the corresponding threshold (exact threshold) functions when all weights are $1$. Slightly abusing notations, we use $\THR,\ETHR,\MAJ,\EMAJ$ to also denote the corresponding classes of functions. We also consider $\AND_n$ and $\OR_n$, with their usual meanings. We use $\DISJOR_n$ to denote the disjoint OR function, that is, an $\OR$ function with the promise that at most one input bit could be true. 

We use $\GAPOR_n$ to denote the gap OR function, that is, an $\OR$ function with the promise that either all inputs are false or at least half of inputs are true. We also use $\SYM$ to denote the class of all symmetric functions. For a $\SYM$ function $C$, we have $C(x) := f\left(\sum\nolimits_{i=1}^{n} x_i\right)$, and we call $f$ as its associated function.

For a class of function like $\THR$, we use $\THR_k$ to denote its sub-class with at most $k$ inputs. For two classes of functions like $\THR$ and $\SYM$, we use $\THR \circ \SYM$ to denote the corresponding class of depth-$2$ circuits. Similar notations are used for more than $2$ classes.

We use $\LT_d$ to denote the depth-$d$ $\THR$ circuit class, that is, $\LT_d := \underbrace{\THR \circ \dotsc \circ \THR}_{\text{$d$ times}}$. Similarly, we use $\HLT_d$ to denote its unweighted version, that is, $\HLT_d := \underbrace{\MAJ \circ \dotsc \circ \MAJ}_{\text{$d$ times}}$.

When we refer to a circuit class without specifying its size, we always assume the size is polynomial. 

\paragraph*{Previous Known Containment Results.} We need the following standard circuit classes containment results for this paper.

\begin{prop}\label{prop:circuit-facts-contain}
	The following holds:
	
	\begin{enumerate}
		\item $\SYM_k \subseteq \DISJOR_k \circ \EMAJ$.
		\item $\THR \subseteq \MAJ \circ \MAJ$~\cite{GoldmannHR92,Hofmeister96}.
		\item $\THR \subseteq \DISJOR \circ \ETHR$~\cite{HansenP10}.		
		\item $\SYM \circ \THR$ and $ \SYM \circ \ETHR $ are contained in $\SYM \circ \MAJ$~\cite{GoldmannHR92,HansenP10}.
		\item $\MAJ \circ \THR$ and $ \MAJ \circ \ETHR $ are contained in $\MAJ \circ \MAJ$~\cite{GoldmannHR92,HansenP10}.
		\item $\SYM \circ \SYM \subseteq \SYM \circ \MAJ$.
		\item $\ETHR \circ \ETHR \subseteq \THR \circ \THR$~\cite{HansenP10}.
		\item $\AND \circ \ETHR \subseteq \ETHR$~\cite{HansenP10}.
		\item $\AND_k \circ \SYM \subseteq \SYM$~\cite{HansenP10} for a constant $k$.
		\item $\EMAJ \subseteq \MAJ \circ \AND_2$~\cite{HansenP10}.
	\end{enumerate}

	Moreover, all statements above have corresponding polynomial-time, deterministic constructions.
\end{prop}
\begin{rem}
	We remark that for Item (4) and (5), only Item (5) is explicitly stated in~\cite{HansenP10}, but it is not hard to see that the technique works equally well with a top $\SYM$ gate.
\end{rem}

We also need the following folklore lemma, which helps us to transform between $\MAJ \circ \AND$ circuits and $\MAJ \circ \OR$ circuits.

\begin{lemma}\label{lm:AND-to-OR}
	Let $x = x_1,x_2,\dotsc,x_k$ be the inputs, there are $k$ $\OR$ functions $O_1,O_2,\dotsc,O_k$ on the inputs (or their negations) such that:
	\[
	\AND(x)  = \left(\sum_{i=1}^{k} O_i(x) \right)  - (k-1).
	\]
\end{lemma}
\begin{proof}
	We define
	\[
	O_i(x) := \left(\bigvee_{j=1}^{i-1} \neg x_j\right) \vee x_i.
	\]
	
	That is, $O_i(x) = 0$ if and only if the first $i-1$ bits are $1$, and the $i$-th bit is $0$. Now, note that if $\AND(x) = 1$, then all bits are $1$, which means all $O_i(x)$'s are $1$. When $\AND(x) = 0$, let $i$ be the index of the first $0$-bit, it is easy to see that $O_i(x) = 0$ and all other $O_j(x)$'s are $1$, and hence $\sum_{i=1}^{k} O_i(x) = k-1$.
\end{proof}


\subsection{Lower Bound From Non-trivial Satisfiability Algorithm}

Here we introduce the algorithm-to-lower bound tools established in a serious works of Williams~\cite{Wil13,Wil14ACC}, and simplified by Ben-Sasson and Viola~\cite{ben2014short}. 

Let $\ckt$ be a circuit class, we use $\ckt_{n}^{s}$ to denote the subset of $\ckt$ with $n$ inputs and size $\le s$. Slightly abusing notation, we also use $\ckt_{n}^s$ to denote the corresponding functions of the circuits in $\ckt_{n}^s$.

We say a circuit class $\ckt$ is efficiently close under projections, if given the description of a circuit $C$ from $\ckt_{n}^{s}$, for indices $i,j \le n$ and a bit $b$, the following functions
\[
\neg C,C(x_1,\dotsc,x_{i-1},x_j \oplus b,x_{i+1},\dotsc,x_n),C(x_1,\dotsc,x_{i-1},b,x_{i+1},\dotsc,x_{n})
\]
belong to $\ckt_{n}^s$, and their corresponding circuit descriptions can be constructed in $\poly(s)$ time.

The following is from~\cite{ben2014short}, we reformulate it a bit for our use here.

\begin{theo}[Theorem 1.5~\cite{ben2014short}]\label{theo:ENP-lowb} Let $s : \mathbb{N} \to \mathbb{N}$ be a growing parameter of $n$, $\ckt$ be efficiently closed under projections and $C_n = \ckt_{n}^{s(n)}$. If the satisfiability of functions $h = g_1 \wedge g_2 \wedge g_3$ where $g_i \in C_{n + O(\log n)}$ is in $\TIME(2^n/n^{\omega(1)})$,
	then there is a function $f$ in $\ENP$ such that $f_n \notin C_n$ for infinitely many $n$'s.
\end{theo}

We also need the following two similar connections with circuit lower bound against $\NEXP$.

\begin{theo}[\cite{ben2014short,Wil13}]\label{theo:NEXP-lowb-poly}
	Let $\ckt$ be efficiently closed under projections. If there is an algorithm solving the satisfiability of functions $h = g_1 \wedge g_2 \wedge g_3$ where $g_i \in C_{n + O(\log n)}^{n^k}$ in $O(2^{n}/n^k)$ time for all $k$, then $\NEXP$ does not have polynomial size $\ckt$ circuits.
\end{theo}

\begin{theo}[\cite{ben2014short,Wil13}]\label{theo:NEXP-lowb-quasi-poly}
Let $\ckt$ be efficiently closed under projections. If there is an algorithm solving the satisfiability of functions $h = g_1 \wedge g_2 \wedge g_3$ where $g_i \in C_{n + O(\log n)}^{s}$ where $s = 2^{\log^k n}$ in $O(2^{n - \log^k n})$ time for all $k$, then $\NEXP$ does not have quasi-polynomial size $\ckt$ circuits.
\end{theo}

\begin{rem}\label{rem:coNTIME-works-prelim}
	We remark that algorithms in both Theorem~\ref{theo:NEXP-lowb-poly} and Theorem~\ref{theo:NEXP-lowb-quasi-poly} can in fact be replaced by co-nondeterministic algorithms with the same running times.
\end{rem}




\section{Structure Lemmas for $\THRTHR$ Circuits}

In this section we present our structure lemmas for $\THRTHR$ circuits, and discuss some applications.

We first need a simple construction, which will be used in both proofs.

\begin{lemma}[Mod $p$ exact threshold gate]\label{lm:mod-p-ETHR}
	Let $G$ be a $\ETHR$ gate with $n$ inputs, $p$ be a prime and $G^{p}$ be the ``mod p'' version of $G$. That is, let $L$ and $T$ be the corresponding linear function and threshold of $G$, $G^{p}(x) := \left[ L(x) \equiv T \pmod{p} \right]$. 
	
	Then $G^{p}$ can be written as a $\DISJOR \circ \ETHR$ circuit such that
	\begin{itemize}
		\item The top $\DISJOR$ gate has $O(n)$ fan-in.
		\item All $\ETHR$ gates have positive weights and thresholds smaller than $O(np)$.
	\end{itemize}
\end{lemma}
\begin{proof}
	Let $w_1,w_2,\dotsc,w_n$ and $T$ be the corresponding weights and threshold of $G$. We reduce each weight $w_i$ in $G$ to $w_i \bmod{p}$, and get another circuit with associate top linear function $L'(x)$. We set $t = T \bmod{p}$, then $L(x) \equiv T \pmod{p}$ is equivalent to $L'(x) = t + k \cdot p$ for some $k \in \{0,1,2, \dotsc, n \}$. Therefore, by enumerating $k$ from $0$ to $n$, we can construct the equivalent $\DISJOR \circ \ETHR$ circuit.
\end{proof}

\subsection{Proof for Structure Lemma I}

We begin with the proof for Structure Lemma I for $\THRTHR$ circuits (restated below).

\begin{reminder}{Lemma~\ref{lm:structure-THRTHR-I}}
	Let $n$ be number of inputs and $s = s(n) \ge n$ be a size parameter.
	Every $s$-size $\THRTHR$ circuit $C$ is equivalent to a $\GAPOR \circ \THR \circ \MAJ$ circuit
	such that:
	
	\begin{itemize}
		\item The top $\GAPOR$ gate has $\poly(s)$ fan-in.
		\item Each sub $\THR \circ \MAJ$ circuit has size $\poly(s)$.
	\end{itemize}
	
	Moreover, the reduction can be computed in \emph{deterministic} $\poly(s)$ time.
\end{reminder}
\begin{proof}
	Let $C'$ be the given $\THRTHR$ circuit. By negating some of its input gates ($\THR$ is closed under negation), we can assume all weights in the top $\THR$ gate of $C'$ are $\le 0$.	By Proposition~\ref{prop:circuit-facts-contain} (3), $C'$ can be transformed into an equivalent $\THR\circ\ETHR$ circuit $C$ of size $t = \poly(s)$.
	
	Let $G_1,G_2,\dotsc,G_{t}$, $w_1,w_2,\dotsc,w_{t}$ be the $\ETHR$ gates on the bottom layer and their corresponding weights in the top gate in $C$. By assumption, we also have all $w_i \le 0$. Let $T$ be the threshold of the top gate in $C$. For all input $x$, we have
	\[
	C(x) = \left[ \sum_{i=1}^{t} w_i \cdot G_i(x) \ge T \right].
	\]
	
	By construction, we can assume that weights in $G_i$ are bounded by $2^{n^{c}}$ for a large constant $c$. Fix an input $x$, let $p$ be a random prime from $2$ to $n^{2c} \cdot t^2 \cdot 10 = \poly(s)$, then with probability at least $1 - 1/10t$, we have $G_i^p(x) = G_i(x)$. Let $C^p$ be the circuit obtained by replacing all $G_i$'s in $C$ by corresponding $G_i^p$'s. 
	
	Then we have: (1) when $C(x) = 1$, by a union bound, with probability at least $1/10$, $C^p(x) = C(x) = 1$. (2) when $C(x) = 0$, note that for all prime $p$, we have $G_i^{p}(x) \ge G_i(x)$ for all $i$, therefore we must have $\sum_{i}^{s} w_i \cdot G_i^p(x) \le \sum_{i}^{s} w_i \cdot G_i(x) < T$ (all $w_i$'s are $\le 0$) and $C^p(x) = 0$. Hence, $C$ is equivalent to a $\GAPOR$ of all $C^p$'s, and by Lemma~\ref{lm:mod-p-ETHR}, each $C^p$ can be written as a $\poly(s)$ size $\THR \circ \MAJ$ circuit, which completes the proof.
\end{proof}

\subsection{Proof for Structure Lemma II}

We next prove Lemma~\ref{lm:structure-THRTHR-II}. The proof consists of two steps, which are specified by Lemma~\ref{lm:weight-reduction-top-THR} and Lemma~\ref{lm:decompose-top-ETHR}.

\begin{lemma}[Weight Reduction at the Top $\THR$ gate]\label{lm:weight-reduction-top-THR}
	Given a $\THR_d \circ \ckt$ circuit (a circuit with a top $\THR$ gate of fan-in $d$) of size $s$, it is equivalent to a $\DISJOR \circ \ETHR \circ \ckt$ circuit such that:
	
	\begin{itemize}
		\item The top $\DISJOR$ gate has $\poly(d)$ fan-in.
		\item Each $\ETHR$ gate has fan-in $d$, whose weights and threshold are positive and smaller than $\poly(d) \cdot 2^n$.
		\item The $\ckt$ part is unchanged.
	\end{itemize}

	The same statement also holds for a $\ETHR_d \circ \ckt$ circuit. Moreover, these reductions can be computed in \emph{randomized} $\poly(s)$ time.
\end{lemma}
\begin{proof}
	
	We only consider the $\THR_d \circ \ckt$ case, the $\ETHR_d \circ \ckt$ case is only simpler.
	
	Let $C$ be the given circuit. First, by Proposition~\ref{prop:circuit-facts-contain} (5), $C$ can be transformed to an equivalent $\DISJOR \circ \ETHR \circ \ckt$ circuit $C'$. 
	
	\newcommand{\Mold}{M_{\textsf{old}}}
	\newcommand{\Mnew}{M_{\textsf{new}}}
	
	Now, we deal with each $\ETHR$ gate $G$ separately, note that $G$ also has fan-in $d$. Let $D$ be the sub-circuit with top gate $G$. From the construction, $G$ may have weight of absolute value at most $\Mold = 2^{\poly(d)}$. 
	
	We next define $L : \{0,1\}^{n} \to \mathbb{Z}$ such that $L(x)$ is the value of the linear function associated with the gate $G$ when the input is $x$. That is $D(x) = 1$ if and only if $L(x) = T$ for the threshold $T$ of $G$.
	
	Then we pick a random prime number $m$ from $0$ to $\Mnew = d^{c} \cdot 2^{n}$, where $c$ is a sufficiently large constant. For a fixed $x \in \{0,1\}^{n}$, if $L(x) \ne T$, the probability that $L(x) \equiv T \pmod{m}$ is smaller than 
	\[
	\frac{\log(\Mold)}{\Mnew / \ln(\Mnew)} = \frac{\poly(d)}{\Theta(2^{n} \cdot d^{c} / (n + c\log d))} \le d^{-c/2} /2^{n},
	\] 
	for a sufficiently large $c$. Therefore, by a simple union bound, with probability at least $1 - d^{-c/2}$, we have $L(x) \equiv T \pmod{m}$ if and only if $L(x) = T$ for all $x \in \{0,1\}^{n}$. We pick such a prime $m$ for gate $G$.
	
	Finally, by applying Lemma~\ref{lm:mod-p-ETHR} with prime $m$, we can replace $G$ with an equivalent $\DISJOR \circ \ETHR$ sub-circuit, whose $\ETHR$ gates have positive weights and thresholds smaller than $\poly(d) \cdot 2^{n}$.
	
	By a union bound over all $\ETHR$ gates, and choose $c$ to be a large enough constant, we complete our randomized reduction.
\end{proof}
\begin{rem}\label{rem:one-sided-error}
	One can observe that the above reduction indeed only introduces \emph{one-sided error}. That is, even it chooses some ``bad'' primes, the resulting circuit $D$ satisfies the property that $D(x) = 1$ whenever $C(x) = 1$.
\end{rem}

\begin{lemma}[Decomposition of the top $\ETHR$ gate]\label{lm:decompose-top-ETHR}
	Given an $\ETHR_d \circ \ckt$ circuit $C$ (a circuit with a top $\ETHR$ gate of fan-in $d$) of size $s$ and a real $\eps \in \left(\frac{\log d}{n},1 \right)$, suppose the top $\ETHR$ gate in $C$ has positive weights and threshold smaller than $2^{2n}$. $C$ is equivalent to a $\DISJOR \circ \MAJ \circ \AND_2 \circ \ckt$ circuit such that:
	
	\begin{itemize}
		\item The top $\DISJOR$ gate has $2^{O(\eps n)}$ fan-in.
		\item Each $\MAJ$ gate has fan-in $d^{O(1/\eps)}$.
		\item The $\ckt$ part is unchanged.
	\end{itemize}

	Moreover, the reduction can be computed in \emph{deterministic} 
	\[
	2^{O(\eps n)} \cdot d^{O(1/\eps)} + \poly(s)
	\] 
	time.
\end{lemma}
\begin{proof}
	\newcommand{\Gtop}{G_{\textsf{top}}}
	Let $\Gtop$ be the top $\ETHR$ in $C$ and $G_1,G_2,\dotsc,G_d$ be its input gates. Let $w_i$'s and $T$ be the weights and the threshold of $\Gtop$ and $L(x)$ be the associated linear function, we have
	\[
	L(x) = \sum_{i=1}^{d} w_i \cdot G_i(x)
	\]
	for all input $x \in \{0,1\}^n$.
	
	Now, note that the binary representations of $w_i$'s and $T$ are of length at most $\log(2^{2n}) \le 2n$, and we break them into $D = \left\lceil \frac{\eps \cdot n}{\log d} \right\rceil$ blocks, each with $ B  \le 2/\eps \cdot \log d$ bits. Let $w_{i,j} \in [2^{B} - 1]$ and $T_j$ be the value of $w_i$'s and $T$'s $j$-th block respectively (blocks are numbered from the least significant bit to the most significant bit).
	
	Consider adding up $w_i \cdot G_i(x)$'s in $2^{B}$ base and enumerate all $D - 1$ carries on each position except for the highest one. Let $c_1,c_2,\dotsc,c_{D-1} \in \{0,1,\dotsc,d-1\}^{D-1}$ be such a carry sequence.	We can see $\sum_{i=1}^{d} w_i \cdot G_i(x) = T$ with respect to a carry sequence $c$ is equivalent to that for all $j \in [D]$: 
	\[
	\sum_{i=1}^{d} w_{i,j} \cdot G_i(x) + c_{j-1} = T_{j} + 2^{B} \cdot c_{j},
	\]	
	where we set $C_D$ and $C_0$ to be $0$ for notational convenience.
	
	That is, after fixing $c_j$'s, for all $j$, $\sum_{i=1}^{d} w_{i,j} \cdot G_i(x)$ are also forced to be $T_{j}^{c} = T_{j} + 2^{B} \cdot c_{j} - c_{j-1}$.
	
	Therefore, consider the sum
	\[
	\sum_{j=1}^{\eps \cdot n} \left(\sum_{i=1}^{d} w_{i,j} \cdot G_i(x) - L_{j}^c\right)^2.
	\]
	
	Checking whether this sum $\le 0$ can be formulated as a $\poly(d) \cdot 2^{O(B)} = d^{O(1/\eps)}$ size $\MAJ \circ \AND_2$ sub-circuit, with input gates $G_1,G_2,\dotsc,G_d$.
	
	Moreover, since each addition process only corresponds to one carry sequence, by enumerate all possible carry sequence, we can see the above transform $\Gtop$ into a $\DISJOR \circ \MAJ \circ \AND_2$ sub-circuit with input gates $G_1,G_2,\dotsc,G_d$, with top fan-in:
	\[
	d^{D-1} = d^{O(\eps \cdot n /\log d)} = 2^{O(\eps \cdot n)},
	\] 
	which completes the proof.
\end{proof}

Finally, Structure Lemma II for $\THRTHR$ circuits follows directly from Lemma~\ref{lm:weight-reduction-top-THR} and Lemma~\ref{lm:decompose-top-ETHR}.

\begin{reminder}{Lemma~\ref{lm:structure-THRTHR-II}}
	Let $n$ be number of inputs and $s = s(n)$ be a size parameter. Let $\eps \in \left(\frac{\log s}{n},1\right)$, for $s = 2^{o(n)}$, every $s$-size $\THRTHR$ circuit $C$ is equivalent to a $\DISJOR \circ \MAJ \circ \MAJ$ circuit such that:
	
	\begin{itemize}
		\item The top $\DISJOR$ gate has $2^{ O(\eps n)}$ fan-in.
		\item Each sub $\MAJ \circ \MAJ$ circuit has size $s^{O(1/\eps)}$.
	\end{itemize}
	
	Moreover, the reduction can be computed in \emph{randomized} $2^{O(\eps n)} \cdot s^{O(1/\eps)}$ time. 
\end{reminder}
\begin{proof}
	By Proposition~\ref{prop:circuit-facts-contain} (3), $C$ is equivalent to a $\poly(s)$ size $\THR \circ \ETHR$ circuit $C_1$.
	
	Then we apply Lemma~\ref{lm:weight-reduction-top-THR} to reduce $C_1$ into a $\DISJOR_{\poly(s)} \circ \ETHR \circ \ETHR$ circuit $C_2$, whose second-layer $\ETHR$ gates have positive weights and thresholds smaller than $\poly(s) \cdot 2^{n} < 2^{2n}$.
	
	Next we apply Lemma~\ref{lm:decompose-top-ETHR} to change all second layer $\ETHR$ gates in $C_2$ into a $\DISJOR \circ \MAJ \circ \AND_2$ sub-circuits, with top gate fan-in $2^{O(\eps \cdot n)}$. Putting everything together, and note that $\AND_2 \circ \ETHR$ can still be represented by an $\ETHR$ gate, we obtain a $\DISJOR_{2^{O(\eps \cdot n)}} \circ \MAJ \circ \ETHR$ circuit, in which all $\MAJ \circ \ETHR$ sub-circuits have size at most $s^{O(1/\eps)}$.
	
	Applying Proposition~\ref{prop:circuit-facts-contain} (5) completes our proof. And the running time bound follows from the corresponding time bounds in Lemma~\ref{lm:weight-reduction-top-THR} and Lemma~\ref{lm:decompose-top-ETHR}.
\end{proof}

The following corollary follows directly by setting the parameter $\eps$ carefully in Lemma~\ref{lm:structure-THRTHR-II}.

\begin{cor}\label{cor:structure-THRTHR-2}
	Let $n$ be number of inputs and $s = s(n)$ be a size parameter. Let $\eps \in \left(\frac{\log s}{n},1\right)$, for $s = 2^{o(n)}$, an $s$-size $\THRTHR$ circuit $C$ is equivalent to a $\DISJOR \circ \MAJ \circ \MAJ$ circuit such that:
	
	\begin{itemize}
		\item The top $\DISJOR$ gate has $s^{O(1/\eps)}$ fan-in.
		\item Each sub $\MAJ \circ \MAJ$ circuit has size $2^{O(\eps \cdot n)}$.
	\end{itemize}
	
	Moreover, the reduction can be computed in \emph{randomized} 
	\[
	2^{O(\eps n)} \cdot s^{O(1/\eps)}
	\] 
	time.
\end{cor}

\subsection{Some Applications}

Finally, we prove these interesting implications of Lemma~\ref{lm:structure-THRTHR-I} and Lemma~\ref{lm:structure-THRTHR-II}. 

The following corollary follows from Lemma~\ref{lm:structure-THRTHR-I} directly.

\begin{reminder}{Corollary~\ref{cor:equivalent-THRTHR-THRMAJ}}
	The following are equivalent:
	\begin{itemize}
		\item The satisfiability of $\THRTHR$ circuits of size $n^k$ can be solved in $2^{n} / n^k$ time for any $k$.
		\item The satisfiability of $\THR\circ\MAJ$ circuits of size $n^k$ can be solved in $2^{n} / n^k$ time for any $k$.
	\end{itemize}
\end{reminder}
\begin{proof}
	We only need to prove the second item implies the first. Suppose the second item holds, given a $\THRTHR$ circuit of size $n^k$, by Lemma~\ref{lm:structure-THRTHR-I}, it can be reduced to an equivalent $\GAPOR \circ \THR \circ \MAJ$ circuit of size $n^{kc}$ for a constant $c$, whose satisfiability can be solved in $2^n / n^{kc}$ time by the first item.
\end{proof}

And the following two corollaries follow from Lemma~\ref{lm:structure-THRTHR-II} directly.

\begin{reminder}{Corollary~\ref{cor:structure-LT}}
	Let $n$ be number of inputs and $s = s(n)$ be a size parameter. Let $\eps \in \left(\frac{\log s}{n},1\right)$ and $d$ be a constant. For $s = 2^{o(n)}$, every $s$-size $\LT_d$ circuit is equivalent to a $\DISJOR \circ \HLT_d$ circuit such that:
	\begin{itemize}
		\item The top $\DISJOR$ gate has $2^{ O(\eps \cdot n) }$ fan-in.
		\item Each sub $\HLT_d$ circuit has size $O\left( s^{O(1/\eps)} \right)$.
	\end{itemize}
\end{reminder}
\begin{proof}
	We apply Lemma~\ref{lm:structure-THRTHR-II} to the top $2$ layers, and then apply Proposition~\ref{prop:circuit-facts-contain} (5) recursively to obtain an equivalent $\DISJOR \circ \HLT_d$ circuit.
\end{proof}

\begin{cor}\label{cor:LT-HLT-equiv}
	For all $d \ge 2$, the following are equivalent:	
	\begin{itemize}
		\item There is a $2^{(1-\Omega(1)) \cdot n}$ time algorithm for satisfiability of polynomial size $\LT_d$ circuits.
		\item There is a $2^{(1-\Omega(1)) \cdot n}$ time algorithm for satisfiability of polynomial size $\HLT_d$ circuits.
	\end{itemize}
\end{cor}
\begin{proof}
	Suppose we have a $2^{(1-\eps_1) n}$ time algorithm for satisfiability of polynomial size $\HLT_d$ circuits for a constant $\eps_1 > 0$. Let $c$ be the hidden constant in the big-$O$ notation of the fan-in of the top $\DISJOR$ gate in Lemma~\ref{lm:structure-THRTHR-II}. 
	
	We set $\eps = \eps_1 / 2c$ and apply Lemma~\ref{lm:structure-THRTHR-II} to the given $\LT_d$ circuit. We obtain an equivalent $\DISJOR \circ \HLT_d$ circuit with top fan-in $2^{c \eps n} = 2^{\eps_1/2 \cdot n} $ and polynomial size $\HLT_d$ sub-circuits. Then we can apply our algorithm for solving polynomial size $\HLT_d$ to solve the satisfiability of the given $\LT_d$ circuit in $2^{(1-\eps_1/2) \cdot n}$ time, which completes the proof.
\end{proof} 

Note that Corollary~\ref{cor:equivalent-THRTHR-MAJMAJ} is simply a special case of the above Corollary when $d = 2$.

Similarly, the same techniques can be used to derive a structure lemma for $\THR \circ \AND_{k}$ circuits as well.

\begin{reminder}{Corollary~\ref{cor:structure-THR-AND}.}
	Let $n$ be number of inputs and $s = s(n)$ be a size parameter. Let $\eps \in \left(\frac{\log s}{n},1\right)$ and $k$ be a constant. Assuming $s = 2^{o(n)}$, an $s$-size $\THR \circ \AND_k$ circuit is equivalent to a $\DISJOR \circ \MAJ \circ \AND_{2k}$ circuit such that:
	
	\begin{itemize}
		\item The top $\DISJOR$ gate has $2^{ O(\eps \cdot n) }$ fan-in.
		\item Each sub $\MAJ \circ \AND_{2k}$ circuit has size $O\left( s^{O(1/\eps)} \right)$.
	\end{itemize}
	
	The above still holds if we replaced both $\AND_{k}$ and $\AND_{2k}$ by unbounded fan-in $\AND$ gates.
\end{reminder}
\begin{proof}
	We simply apply Lemma~\ref{lm:weight-reduction-top-THR} and Lemma~\ref{lm:decompose-top-ETHR}, and merge each $\AND_2 \circ \AND_{k}$ sub-circuits into a single $\AND_{2k}$ gate.
\end{proof}

Together with Lemma~\ref{lm:AND-to-OR}, the following corollary is evident.

\begin{reminder}{Corollary~\ref{cor:weight-unweight-MAXSAT-equiv}.}
	For any integer $k$, if there is a $2^{(1-\Omega(1)) n}$ time algorithm for polynomial size unweighted $\MAXkSAT[2k]$, then so does polynomial size weighted $\MAXkSAT$.
\end{reminder}
\begin{proof}
	We can use Lemma~\ref{lm:AND-to-OR} to transform the bottom $\AND$ gates to $\OR$ gates for $\THR\circ\AND$ and $\MAJ\circ\AND$ circuits, and then the proof are exactly the same as in Corollary~\ref{cor:LT-HLT-equiv}.
\end{proof}


\section{Shaving Logs from $\ell_2$-Furthest Pair Implies $\THRTHR$ Lower Bound}

In this section we show shaving logs on $\ell_2$-Furthest Pair or other related problems would have exciting circuit lower bound consequence. 

We first show that slightly faster satisfiability algorithm for $\THRTHR$ implies circuit lower bound against $\THR \circ \THR$. Note that this is not obvious as $\THRTHR$ circuits are not trivially closed under intersection, while we have to solve satisfiability for an $\AND$ of $3$ $\THRTHR$ circuits faster.

\begin{lemma}\label{lm:lowb-THRTHR}
	If there is an algorithm solving the satisfiability of $\THRTHR$ circuits of size $n^k$ in $2^{n} / n^k$ time for any $k$, then $\NEXP$ has no polynomial size $\THRTHR$ circuits.
\end{lemma}
\begin{proof}
	From Theorem~\ref{theo:NEXP-lowb-poly}, we have to devise a $2^{n} / \log^{\omega(1)} n$ time algorithm for solving $\AND_3 \circ \THR \circ \THR$ circuits of size $s = n^{k}$ with $n' = n + O(\log n)$ inputs.
	
	Given such a circuit $C$, we first apply Proposition~\ref{prop:circuit-facts-contain} (3) to transform it into a $\poly(s)$ size $\AND_3 \circ \DISJOR \circ \ETHR \circ \ETHR$ circuit $C'$.
	
	Note that we can switch the order of $\DISJOR$ and $\AND_3$, by treating the first as addition and the second as multiplication. Then $C'$ is equivalent to another $\DISJOR \circ \ETHR \circ \ETHR$ circuit $C''$ of $\poly(s)$ size.
	
	Finally, solving $C''$ can be completed by solving $\poly(s)$ $\ETHR \circ \ETHR$ sub-circuits, and note that $\ETHR \circ \ETHR \subseteq \THR \circ \THR$ (Proposition~\ref{prop:circuit-facts-contain} (7)), hence using the algorithm from the assumption completes the proof.
\end{proof}
\begin{rem}\label{rem:co-NTIME-work-too}
	By Remark~\ref{rem:coNTIME-works-prelim}, the consequence also holds if the algorithm in Lemma~\ref{lm:lowb-THRTHR} is co-nondeterministic.
\end{rem}

\newcommand{\LEQ}{\textsf{LEQ}}

\begin{lemma}\label{lm:THRTHR-to-WMaxIP}
	If there is an algorithm solving $\WMaxIP_{n,\log^k n}$ or $\Hopcroft_{n,\log^k n}$ in $n^2 / \log^k(n)$ time for any integer $k$, then $\NEXP$ has no polynomial size $\THRTHR$ circuits.
\end{lemma}
\begin{proof}
	We consider $\WMaxIP$ first. Suppose there is such an algorithm for $\WMaxIP$. By Corollary~\ref{cor:equivalent-THRTHR-THRMAJ} and Lemma~\ref{lm:lowb-THRTHR}, we only need to devise an algorithm for the satisfiability of $\THR\circ\MAJ$ circuits of size $n^k$ in $2^n / n^k$ time for all $k$. We do so by reducing the satisfiability problem for $\THR \circ \MAJ$ circuits to $\WMaxIP$ or $\Hopcroft$.
		
	For simplicity, we assume $n$ is even. Let $C$ be a $\THR\circ\MAJ$ circuit of size $s = n^k$ and $G$ be its top $\THR$ gate. Let $W_1,W_2,\dotsc,W_{s}$, $T$ and $L$ be the weights, threshold and associate linear function of $G$. Let $G_1,G_2,\dotsc,G_{s}$ be the corresponding $\MAJ$ gates on the bottom layers. We use $L_1,L_2,\dotsc,L_{s}$ and $T_{1},T_{2},\dotsc,T_{s}$ to denote their associated linear functions and thresholds.
	
	For each $x,y \in \{0,1\}^{n/2}$, we interpret $x$ and $y$ as an assignment to the first half and second half of the input to $C$ respectively.
	
	For each linear functions $L_{j}$, we use $X_j(x)$ and $Y_j(y)$ to denote the contribution from $x$ and $y$ respectively. We have
	\[
	G_{j}(x,y) := \left[ X_j(x) + Y_j(y) \ge T_j \right].
	\]
	
	Note that since each $G_j$ has at most $s$ wires, and therefore $0 \le X_j(x), Y_j(y) \le s$. So we now define $u_j(x),v_j(y) \in \{0,1\}^{s + 1}$, such that $(u_j(x))_i = 1$ iff $i = X_j(x)$ and $(v_j(y))_i = 1$ iff $i + Y_j(y) \ge T_j$. Then we have $G_j(x,y) = u_j(x) \cdot v_j(y) $. Now we set
	\[
	u(x) := \circ_{j=1}^{s} u_j(x) \quad v(y) := \circ_{j=1}^s v_j(x) \quad w := \circ_{j=1}^s W_j^{\otimes(s+1)}.
	\]
	
	It is easy to see that $u(x) \odot_{w} v(y) = L(x,y)$. Therefore, computing the maximum of $u(x) \odot_w v(y)$ for all $x,y \in \{0,1\}^{n/2}$ solves the problem, which can be reduced to a $\WMaxIP_{2^{n/2},\poly(n)}$ instance. The proof is completed by applying the algorithm for $\WMaxIP$ in the assumption.
	
	The reduction to $\Hopcroft$ works roughly the same, with the only modification that we transform the $\THR\circ\MAJ$ circuit into an equivalent $\DISJOR\circ\ETHR\circ\MAJ$ at the beginning (via Proposition~\ref{prop:circuit-facts-contain} (3)), and solve each $\ETHR \circ \MAJ$ sub-circuits separately via a similar reduction to $\Hopcroft$.
\end{proof}

Now we are ready to prove Theorem~\ref{theo:Hopcroft-THRTHR-NEXP}.

\begin{reminder}{Theorem~\ref{theo:Hopcroft-THRTHR-NEXP}}
	If any of the following problems has an $n^{2} \poly(d) / \log^{\omega(1)} n$ time deterministic algorithm for polylogarithmic $d$, then $\NEXP$ has no polynomial size $\THR \circ \THR$ circuits:
	
	\begin{enumerate}
		
		\item $\Hopcroft_{n,d}$ (Hopcroft's Problem): Find an orthogonal pair among $n$ points in $\mathbb{Z}^d$.
		
		\item $\FPtwo_{n,d}$:  Find the $\ell_2$-furthest pair among $n$ points in $\mathbb{R}^d$.
		
		\item $\BCPtwo_{n,d}$: Given two set $A,B$ of $n$ points in $\mathbb{R}^{d}$, compute $\min_{(a,b) \in A \times B} \|a - b\|_2$.
		
		\item $\IntMaxIP_{n,d}$: Given two sets $A,B$ of $n$ vectors from $\mathbb{Z}^{d}$, compute $\max_{(a,b) \in A \times B} a \cdot b$.
		
		\item $\WMaxIP_{n,d}$: Given a weight vector $w \in \mathbb{Z}^{d}$ and two sets $A,B$ of $n$ vectors from $\{0,1\}^{d}$, compute $\max_{(a,b) \in A \times B} \sum_{i=1}^{d} w_i \cdot a_i \cdot b_i $.
	\end{enumerate}
\end{reminder} 

\begin{proof}
	The cases of $\WMaxIP_{n,d}$ and $\Hopcroft_{n,d}$ follow directly from Lemma~\ref{lm:THRTHR-to-WMaxIP}. 
	
	And the cases for other problems follow from the fact there are efficient reductions from $\Hopcroft_{n,d}$ to all of them~\cite{Wil18} (see also Theorem~4.3, Lemma~4.5 and Lemma~4.6 of~\cite{Che18} for explicit reductions).
\end{proof} 
\section{Shaving Logs from Approximate $\BCPtwo$ Implies $\SYMTHR$ Lower Bound}

In this section we establish circuit lower bound consequences from shaving logs on Approximate \BCPtwo\ or other related problems.

We need the following Lemma first, whose proof is deferred to the end of this section.

\begin{lemma}\label{lm:SYMSYM-to-Max-IP}
	Given a size $s$ $\AND_3 \circ \SYM \circ \SYM$ circuit $C$, there is a $2^{n/2} \poly(s)$ time algorithm reducing it into $s^{3}$ $\MaxIP_{2^{n/2},O(s^2 n^2)}$ instances.
\end{lemma}

To prove Theorem~\ref{theo:Max-IP-NEXP}, we first show the following reductions from $\MaxIP$.

\begin{lemma}\label{lm:simple-lemma-1}
	Let $n,d$ be two integers and $\eps = 1/10d$, a $\MaxIP_{n,d}$ instances can be reduced to:
	\begin{itemize}
		\item $(1+\eps)$-approximation to $\BCPone_{n}$.
		\item $(1+\eps)$-approximation to $\BCPtwo_{n}$.
	\end{itemize}
\end{lemma}
\begin{proof}
	Given a $\MaxIP_{n,d}$ instance with two sets $A, B \subseteq \{0,1\}^d$. We first consider Item (2). For each $x \in A$ and $y \in B$, we create two points $p_x$ and $q_y$ in $\mathbb{R}^{d+2}$, such that
	\[
	p_x = \left(x,\sqrt{d - \|x\|_2^2},0\right) \text{, } q_y = \left(y,0,\sqrt{d - \|y\|_2^2}\right).
	\]
	
	We have
	\[
	\|p_x - q_y\|_2^2 = \|p_x\|_2^2 + \|q_y\|_2^2 - 2 \cdot (p_x \cdot p_y) = 2 \cdot ( d - x \cdot y ).
	\]
	
	Note that a $(1+\eps)$-approximation to $\min_{(x,y) \in A \times B} \|p_x - q_y\|_2$ imply a $(1+\eps)^2$-approximation to $\min_{(x,y) \in A \times B} \|p_x - q_y\|_2^2 =  \min_{(x,y) \in A \times B} 2 \cdot (d - x \cdot y)$. Note that $\eps = 1/10d$, we can determine $\max_{(x,y) \in A \times B} x \cdot y$ from the above approximation immediately, which completes the reduction to Item (2).
	
	For Item (1), we begin by setting up some notations. For $t \in [d]$, we use $e^{[t]}$ to denote the Boolean vector with first $t$ coordinates being $1$, and the rest being $0$. Recall that for two vectors $a,b$, we use $a \circ b$ to denote their concatenation.
	
	For each $x \in A$ and $y \in B$, we create two points $p_x,q_y \in \{0,1\}^{3d}$, such that
	\[
	p_x = x \circ e^{[d - \|x\|_1]} \circ e^{[0]}, q_y = y \circ e^{[0]} \circ e^{[d - \|y\|_1]}.
	\]
	
	Note that for each $p_x$ and $q_y$, there are exactly $d$ coordinates with value $1$. Also, note that their inner product $p_x \cdot q_y$ corresponds to the number of coordinates on which they are both $1$. We have
	\[
	\|p_x - q_y\|_1 = 2 \cdot (d - p_x \cdot p_y).
	\]
	
	Therefore, a $(1+\eps)$-approximation to $\min_{(a,b) \in A \times B} \|p_x - q_y\|_1 = \min_{(a,b) \in A \times B} 2 \cdot (d - p_x \cdot p_y)$ would be enough to solve the given $\MaxIP$ instance, which complete the proof for Item (1).
\end{proof}

Now we are ready to prove Theorem~\ref{theo:Max-IP-NEXP} (restated below).

\begin{reminder}{Theorem~\ref{theo:Max-IP-NEXP}}
	If any of following problems has an $n^{2} \poly(d) / \log^{\omega(1)} n$ time deterministic algorithm for polylogarithmic $d$, then $\NEXP$ has no polynomial size $\SYM \circ \THR$ circuits:
	
	\begin{enumerate}
		\item $\MaxIP_{n,d}$: Given two sets $A,B$ of $n$ vectors from $\{0,1\}^{d}$, compute $\max_{(a,b) \in A \times B} a \cdot b$.
		
		\item Compute a $(1 + 1/\log^{\omega(1)} n)$-approximation to $\BCPtwo_{n}$.
		
		\item Compute a $(1 + 1/\log^{\omega(1)} n)$-approximation to $\BCPone_{n}$.
	\end{enumerate}
\end{reminder}
\begin{proof}
	By Lemma~\ref{lm:simple-lemma-1}, we only need to consider Item (1) here. Note that by Proposition~\ref{prop:circuit-facts-contain} (4), we just need to consider polynomial size $\SYM\circ\SYM$ circuit.
	
	By Theorem~\ref{theo:NEXP-lowb-poly}, we need to show the satisfiability problem for polynomial size $ \AND_3 \circ \SYM\circ\SYM$ circuits with $n + O(\log n)$ inputs can be solved in $2^n / n^{\omega(1)}$ time. With Lemma~\ref{lm:SYMSYM-to-Max-IP}, it can be reduced to polynomial many $\MAX_{2^{n/2 + O(\log n)},\poly(n)}$ instance, apply our algorithm from Item (1), the needed $2^n / n^{\omega(1)}$ time algorithm follows directly.
\end{proof}

Finally, we prove Theorem~\ref{theo:Max-IP-ENP}, which gives more refined circuit lower bounds consequences.

\begin{reminder}{Theorem~\ref{theo:Max-IP-ENP}}
	Suppose for some a real $k > 2$, one of the following algorithms exists:
	
	\begin{enumerate}
		\item An $n^2 / \log^{\omega(1)} n$ time algorithm for $\MaxIP_{n, \log^k n}$.
		
		
		\item A $(1 + 1/\log^k n)$-approximation algorithm for $\BCPone_{n}$ in $n^2 / \log^{\omega(1)} n$ time.
		
		\item A $(1 + 1/\log^k n)$-approximation algorithm for $\BCPtwo_{n}$ in $n^2 / \log^{\omega(1)} n$ time.
	\end{enumerate}
	
	Then $\ENP$ has no $n^{(k-2)/2 - \eps_1}$ size $\SYM \circ \SYM$ circuit for any $\eps_1 > 0$.
\end{reminder}
\begin{proof}
	Let $\eps_1 > 0$, by Theorem~\ref{theo:ENP-lowb}, it suffices to show that the satisfiability of $s = n^{(k-2)/2 - \eps_1}$ size $\AND_3 \circ \SYM \circ \SYM$ circuits with $n' = n + O(\log n)$ inputs can be solved in $2^{n} / n^{\omega(1)}$ time.
	
	We consider Item (1) first. By Lemma~\ref{lm:SYMSYM-to-Max-IP}, in $2^{n'/2} \poly(s) = 2^{n/2} \poly(s,n)$ time, the aforementioned problem can be reduced to $s^{3}$ instances of $\MAX_{2^{n'/2},O(s^2n'^2)}$. Note that $s^2 n'^2 \le n^{k - \eps_1}$.
	
	Therefore, applying the algorithm for $\MaxIP_{n,c \log n}$, these $s^3$ instances of $\MAX_{2^{n'/2},n^{k - \eps_1}}$ can be solved in
	\[
	s^3 \cdot \left( 2^{n/2} \cdot \poly(n) \right)^2 / n^{\omega(1)} = 2^{n} / n^{\omega(1)}
	\]
	time, which completes the proof for Item (1). Applying Lemma~\ref{lm:simple-lemma-1} and proceed similarly, the claim for the other two cases can also be established.
	
	
\end{proof}

\subsection{Proof of Lemma~\ref{lm:SYMSYM-to-Max-IP}}
To prove Lemma~\ref{lm:SYMSYM-to-Max-IP}, we introduce two simple lemmas first.

\begin{lemma}\label{lm:encoding-trick}
	There are two functions $\psirevx,\psirevy : \{0,1\}^* \to \{0,1\}^*$ such that for all integer $d$ and $x,y \in \{0,1\}^d$, we have $\psirevx(x),\psirevy(y) \in \{0,1\}^{2d}$ and $\psirevx(x) \cdot \psirevy(y) = d - x \cdot y$.
\end{lemma}

\begin{proof}
	We define two functions $\varphi_x,\varphi_y : \{0,1\} \to \{0,1\}^2$ such that:
	\[
	\varphi_x(0) := (1,0),\quad \varphi_x(1) := (0,1),\quad \varphi_y(0) := (1,1),\quad \varphi_y(1) := (1,0).
	\]
	
	It is easy to check that for $a,b \in \{0,1\}$, $a \cdot b = 1 - \varphi_x(a) \cdot \varphi_y(b)$. Then, for $x,y \in \{0,1\}^{d}$, we define $\varphi_x(x) \in \{0,1\}^{2d}$ as the concatenation of $\varphi_x(x_i)$ for each $i \in [d]$, and similarly define $\varphi_y(y) \in \{0,1\}^{2d}$ as the concatenation of $\varphi_y(y_i)$ for each $i \in [d]$.
	
	Then we can see $\psirevx(x) \cdot \psirevy(y) = \sum_{i=1}^{d} \varphi_x(x_i) \cdot \varphi_y(y_i) =  d - x \cdot y$.
\end{proof}

\begin{lemma}\label{lm:micro-reduction-MIN}
	For all integers $d$ and $0 \le m \le d$, there are two mappings $\varphi^x_{d,m},\varphi^y_{d,m} : \{0,1\}^{d} \to \{0,1\}^{O(d^2)}$ and an integer $M_{d,m}$, such that for all $x,y \in \{0,1\}^d$:
	
	\begin{itemize}
		\item If $x \cdot y = m$, then $\varphi^x_{d,m}(x) \cdot \varphi^y_{d,m}(y) = M_{d,m}$.
		\item Otherwise, $\varphi^x_{d,m}(x) \cdot \varphi^y_{d,m}(y) > M_{d,m}$.
	\end{itemize}
	
\end{lemma} 
\begin{proof}
	We remark the reduction here is essentially the same as the trick used in~\cite{Wil18}. For a vector $v \in \{0,1\}^*$, we use $v^{\otimes k}$ to denote the concatenation of $k$ copies of $v$. 
	
	Consider the following polynomial $P(x,y) := (x \cdot y - m)^2$, we have
	\[
	P(x,y) = (x \cdot y)^2 - 2  m (x \cdot y) + m^2 = (x \cdot y)^2 + 2  m (d - x \cdot y ) + m^2 - 2dm.
	\]
	
	For $x,y \in \{0,1\}^{d}$, we construct $\WT{x},\WT{y} \in \{0,1\}^{d^2}$ such that $\WT{x}_{i} = x_{ \lfloor (i-1) /d \rfloor + 1 }$ and $\WT{y}_{i} = - y_{(i \bmod{d}) + 1}$. Then we can see $\WT{x} \cdot \WT{y} = \sum_{i=1}^{d}\sum_{j=1}^d x_i \cdot y_j = (x \cdot y)^2$.  Let $\psirevx$ and $\psirevy$ be the two functions from Lemma~\ref{lm:encoding-trick}. For $x,y \in \{0,1\}^d$, we define
	\[
	\varphi^x_{d,m}(x) := (\WT{x},\psirevx(x)^{\otimes (2m)}) 
	\qquad\text{ and }\qquad 
	\varphi^y_{d,m}(y) := (\WT{y},\psirevx(y)^{\otimes (2m)}).
	\]
	
	Then we have $\varphi^x_{d,m}(x) \cdot \varphi^y_{d,m}(y) = (x \cdot y)^2 + 2  m (d - x \cdot y ) = P(x,y) + 2dm - m^2$. And we set $M_{d,m} = 2dm - m^2$.
	
	Now, if $x \cdot y = m$, we have $P(x,y) = 0$, and therefore $\varphi^x_{d,m}(x) \cdot \varphi^y_{d,m}(y) = M_{d,m}$. Otherwise, $x \cdot y \ne m$ and we have $P(x,y) > 0$, and hence $\varphi^x_{d,m}(x) \cdot \varphi^y_{d,m}(y) > M_{d,m}$. 
	
	Finally, note that $\varphi^x_{d,m}(x),\varphi^y_{d,m}(y) \in \{0,1\}^{d^2 + 2dm}$, which completes the proof.
\end{proof}

The following corollary follows directly from composing the reductions in Lemma~\ref{lm:micro-reduction-MIN} and Lemma~\ref{lm:encoding-trick}.
\begin{cor}\label{cor:micro-reduction-MAX}
	For all integers $d$ and $0 \le m \le d$, there are two mappings $\varphi^x_{d,m},\varphi^y_{d,m} : \{0,1\}^{d} \to \{0,1\}^{O(d^2)}$ and an integer $M_{d,m}$, such that for all $x,y \in \{0,1\}^d$:
	
	\begin{itemize}
		\item If $x \cdot y = m$, then $\varphi^x_{d,m}(x) \cdot \varphi^y_{d,m}(y) = M_{d,m}$.
		\item Otherwise, $\varphi^x_{d,m}(x) \cdot \varphi^y_{d,m}(y) < M_{d,m}$.
	\end{itemize}	
\end{cor} 

Now we are ready to prove Lemma~\ref{lm:SYMSYM-to-Max-IP} (restated below).

\begin{reminder}{Lemma~\ref{lm:SYMSYM-to-Max-IP}.}
	Given an $\AND_3 \circ \SYM_{s} \circ \SYM$ circuit $C$, there is a $2^{n/2} \poly(s)$ time algorithm reducing it into $s^{3}$ $\MaxIP_{n,O(s^2 n^2)}$ instances.
\end{reminder}
\begin{proof}
	For simplicity, we assume $n$ is even throughout the proof. By Proposition~\ref{prop:circuit-facts-contain} (1),we can transform $C$ into an $\AND_3 \circ \DISJOR_{s} \circ \EMAJ_s \circ \SYM $ circuit $C'$, which can be in turn transformed into a $\DISJOR_{s^3} \circ \AND_3 \circ \EMAJ_s \circ \SYM$ circuit $C''$.
	
	Then, for each $\AND_3 \circ \EMAJ_s \circ \SYM$ sub-circuit $D$ of $C''$, we reduce it into a $\MaxIP_{n,O(s^2 n^2)}$ instance. For $j \in [3]$, let $D_{j}$ be the $j$-th $\EMAJ_s \circ \SYM$ sub-circuit of $D$, and let $G_1,G_2,\dotsc,G_s$ be all the $s$ $\SYM$ gates in $D_{j}$, and let $f_1,f_2,\dotsc,f_{s}$ be their corresponding functions. Let $T_j$ be the threshold of the top $\EMAJ$ gate of $D_j$.
	
	For each $x,y \in \{0,1\}^{n/2}$, we interpret $x$ and $y$ as an assignment to the first half and second half of the input to $D$ respectively. We use $X_{i}(x)$ and $Y_i(y)$ to denote the contribution of $x$ and $y$ to gate $G_i$ respectively. Then we have
	\[
	G_i(x,y) = f_i(X_i(x) + Y_i(y)).
	\]
	
	Now, for an integer $t \in \{0,1,\dotsc,n\}$ and a function $f : \{0,1,\dotsc,n\} \to \{0,1\}$, we define two mappings $\psi_x^f(t),\psi_y^f(t) \in \{0,1\}^{n}$, such that
	\[
	\psi_x^f(t)_i = \begin{cases}
	1 &\quad \text{$i = t$}\\
	0 &\quad \text{otherwise}
	\end{cases}
	\quad\text{and}\quad
	\psi_y^f(t)_i = \begin{cases}
	1 &\quad \text{$f(i+t) = 1$}\\
	0 &\quad \text{otherwise.}
	\end{cases}
	\]
	
	Then we can see for two integers $a,b \in \{0,1,\dotsc,n\}$, $\psi_x^f(a) \cdot \psi_y^f(b) = f(a+b)$.
	
	Now, for each $x,y \in \{0,1\}^{n/2}$ ,we define
	\[
	\psi^j_x(x) := \circ_{i=1}^{s} \psi_x^{f_i}(X_i(x)) \quad\text{and}\quad \psi^j_y(y) := \circ_{i=1}^{s} \psi_y^{f_i}(Y_i(y)).
	\]
	
	Therefore, we have
	\[
	\psi^j_x(x) \cdot \psi^j_y(y) = \sum_{i=1}^{s} \psi_x^{f_i}(X_i(x)) \cdot \psi_y^{f_i}(Y_i(y)) = \sum_{i=1}^{s} G_i(x,y),
	\]
	and consequently
	\[
	D_j(x,y) = \left[ \psi^j_x(x) \cdot \psi^j_y(y) = T_j \right].
	\]
	
	Note that $\psi^j_x(x),\psi^j_y(y) \in \{0,1\}^{sn}$. In order to compute the $\AND$ of $D_1$,$D_2$ and $D_3$, we make use of Corollary~\ref{cor:micro-reduction-MAX}, consider
	
	\[
	\psi_x(x) := \circ_{j=1}^{3} \left( \varphi^{x}_{sn,T_j} (\psi^j_x(x)) \right) \quad\text{and}\quad \psi_y(y) := \circ_{j=1}^{3} \left( \varphi^y_{sn,T_j}(\psi^j_y(y)) \right).
	\]
	
	Let $M =\sum_{j=1}^{3} M_{sn,T_j} $. From Corollary~\ref{cor:micro-reduction-MAX}, note that $\psi_x(x) \cdot \psi_y(y) = M$ if $D_1(x,y) \wedge D_2(x,y) \wedge D_3(x,y)$, and $\psi_x(x) \cdot \psi_y(y) < M$ otherwise. Therefore, let $A$ be the set of all $\psi_x(x)$'s for $x \in \{0,1\}^{n/2}$, and $B$ be the set of all $\psi_y(y)$'s for $y \in \{0,1\}^{n/2}$. We can see $A,B$ form a $\MaxIP_{n,O(s^2 n^2)}$ instance and $\MAX(A,B) = M$ if and only if $D$ is satisfiable.
	
	Therefore, by reducing all $O(s^3)$ $\AND_3 \circ \EMAJ_s \circ \SYM$ sub-circuits of $C''$ into $\MaxIP_{n,O(s^2 n^2)}$ instances, we solve the satisfiability problem for the equivalent $\AND_3 \circ \SYM_s \circ \SYM$ circuit $C$. This completes the proof.
\end{proof}
\section{Shaving Logs from Modest Dimension $\MaxIP$ Implies $\THRTHR$ Lower Bound}

In this section we prove Theorem~\ref{theo:Max-IP-NEXP} (restated below).

\begin{reminder}{Theorem~\ref{theo:MaxIP-THRTHR-NEXP}}
	If any of the following deterministic algorithms exists, then $\NEXP$ has no polynomial-size $\THRTHR$ circuits:
	
	\begin{enumerate}
		\item An algorithm solving $\MaxIP_{n,n^\eps}$ in $n^{2} / \log^{\omega(1)}(n)$ time, for a constant $\eps > 0$.
		
		\item An algorithm solving $\MaxIP_{n,\log^k(n)}$ in $n^{2-\eps}$ time for a constant $\eps > 0$ and any integer $k$.
	\end{enumerate}
\end{reminder}
\begin{proof}
	We first consider Item (2). We want to apply Lemma~\ref{lm:structure-THRTHR-II} to simplify the given $\THRTHR$ circuit. However, the problem here is that Lemma~\ref{lm:structure-THRTHR-II} only implies a \emph{randomized} reduction, preventing us from applying Lemma~\ref{lm:lowb-THRTHR}, as that needs a \emph{deterministic} algorithm.
	
	Fortunately, by Remark~\ref{rem:co-NTIME-work-too}, we only need to come up with a co-nondeterministic algorithm. That is, we want a nondeterministic algorithm which decides whether a $\THRTHR$ circuit of size $n^k$ is unsatisfiable in $2^{n} / n^k$ time for every integer $k$.
	
	In the following we derandomize the construction in Lemma~\ref{lm:structure-THRTHR-II} using nondeterminism. Given a $\THRTHR$ circuit $C$ of size $s = n^k$, we also construct its negation $D = \neg C$, with the same size $s$.
	
	Let $\eps_1 > 0$ be a constant to be specified later. We apply the reduction of Lemma~\ref{lm:structure-THRTHR-II} to both $C$ and $D$, and guess all random primes needed nondeterministically alone the way, which takes 
	\[
	2^{O(\eps_1n)} \cdot s^{O(1/\eps_1)}
	\]
	nondeterministic time.
	
	After that, we get two $\DISJOR \circ \MAJ \circ \MAJ$ circuits $C'$ and $D'$, whose top $\DISJOR$ gates have fan-in $2^{O(\eps_1 n)} \cdot \poly(s) = 2^{O(\eps_1 \cdot n)} $, and each sub $\MAJ \circ \MAJ$ circuit has size $s^{O(1/\eps_1)}$. We have to verify that $C$ and $D$ are indeed equivalent to $C'$ and $D'$. We claim that holds if and only if $C' \wedge D'$ are unsatisfiable.
	
	One direction is straightforward. If they are equivalent correspondingly, then $D'$ is the negation of $C'$ too, and $C' \wedge D'$ are unsatisfiable.

	For the other direction, note that the reduction of Lemma~\ref{lm:structure-THRTHR-II} only introduces \emph{one-sided} error (Remark~\ref{rem:one-sided-error}). That is, for all possible guess and $x \in \{0,1\}^{n}$, when $C(x) = 1$, we must have $C'(x) = 1$. And the same holds for $D(x)$ and $D'(x)$. Therefore, suppose $C$ is not equivalent to $C'$ (the case for $D$ and $D'$ is similar), it must be the case that there is an $x$ such that $C(x) = 0$ while $C'(x) = 1$. Since $C(x) = 0$, we have $D(x) = 1$ and therefore $D'(x) = 1$, which means $(C ' \wedge D')(x) = 1$, completes the proof of the claim.
	
	Note that $C' \wedge D'$ is an $\AND_2 \circ \DISJOR \circ \MAJ \circ \MAJ$ circuit. We can switch the order of $\AND_2$ and $\DISJOR$ by treating them as multiplication and addition respectively, and obtain an equivalent $\DISJOR \circ \AND_2 \circ \MAJ \circ \MAJ$ circuit, with top-fan in $2^{O(\eps_1 \cdot n)}$ and sizes of its $\MAJ \circ \MAJ$ sub-circuits unchanged.
	
	Applying Lemma~\ref{lm:SYMSYM-to-Max-IP}, the satisfiability of an $\AND_2 \circ \MAJ \circ \MAJ$ circuit can be reduced to $\poly(s^{O(1/\eps_1)}) = n^{O(k/\eps_1)}$ $\MaxIP_{2^{n/2},n^{O(k/\eps_1)}}$ instances. Therefore, by choosing $\eps_1$ small enough comparing to $\eps$, we can obtain a $2^{(1-\eps/2) \cdot n}$ time algorithm for the satisfiability of $C' \wedge D'$ from the algorithm in Item (2).
	
	Finally, we reject immediately if we find $C' \wedge D'$ is satisfiable. Otherwise, we know $C'$ is equivalent to $C$, using the same argument we can obtain a $2^{(1-\eps/2) \cdot n}$ time algorithm for the satisfiability of $C'$. We accept only if $C'$ is unsatisfiable.
	
	It is not hard to see the above algorithm solves the unsatisfiability problem of $\THRTHR$ circuits of size $n^k$ in $2^{(1-\eps/2) \cdot n}$ nondeterministic time for any integer $k$, which completes the proof.
	
	The case for Item (1) are roughly the same, except for that we apply Corollary~\ref{cor:structure-THRTHR-2} instead.
\end{proof}

\section{$\MAXSAT$}


In this section we show that slightly better exact algorithms for $\MAXSAT$ would have interesting circuit lower bound consequences. 

To prove Theorem~\ref{theo:MAX-SAT-NEXP} and Theorem~\ref{theo:MAX-SAT-ENP}, we first establish a simple reduction from a $\SYMAND$ circuit to an equivalent $\OR\circ\MAJ\circ \OR$ circuit. 

\begin{lemma}\label{lm:reduction}
	A size $s$ $\SYMAND$ circuit is equivalent to a $\poly(s)$ size $\OR \circ \MAJ \circ \OR$ circuit. Moreover, the latter circuit can be constructed in $\poly(s)$ time.
\end{lemma}
\begin{proof}
	Let $C'$ be the given $\SYMAND$ circuit of size $s$. We can first transform $s$ to an equivalent $\poly(s)$ size $\OR \circ \EMAJ \circ \AND$ circuit $C$ with top fan-in $s$, since $\SYM_s \subseteq \OR_s \circ \EMAJ_s$ (Proposition~\ref{prop:circuit-facts-contain} (1) ).
	
	Now, let $C_1,C_2,\dotsc,C_s$ be all the $\EMAJ \circ \AND$ sub-circuits of $C$. Since $\EMAJ_s \subseteq \MAJ_{O(s^2)} \circ \AND_2$ (Proposition~\ref{prop:circuit-facts-contain} (10)), each $C_i$ is consequently equivalent to a $\poly(s)$ size $\MAJ \circ \AND$ circuit $D_i$.
	
	Let $E_1,E_2,\dotsc,E_{t}$ be all the $\AND$ gates in $D_i$, by Lemma~\ref{lm:AND-to-OR}, supposing $E_{j}$ acts on $k$ variables, we can construct $k$ $\OR$ gates $O_1,O_2,\dotsc,O_k$, such that
	\[
	\sum_{i=1}^{k} O_{i}(x) = (k-1) + E_j(x). 
	\]
	
	Therefore, $D_i$ can be reduced to an equivalent $\MAJ \circ \OR$ circuit, and the proof is completed.
\end{proof}

Now we are ready to prove Theorem~\ref{theo:MAX-SAT-NEXP} and Theorem~\ref{theo:MAX-SAT-ENP} (restated below).

\begin{reminder}{Theorem~\ref{theo:MAX-SAT-NEXP}}
	If there is an algorithm for $\MAXSAT$ solving an instance with $2^{\log^k n}$ clauses in $2^{n} / 2^{\log^k n}$ time for every integer $k$. Then $\NEXP$ has no quasi-polynomial size $\SYMAND$ circuit.
\end{reminder}
\begin{proof}
	By Theorem~\ref{theo:NEXP-lowb-quasi-poly}, it suffices to show that the satisfiability of $2^{\log^k n}$ size $\AND_3 \circ \SYMAND$ circuits with $n+O(\log n)$ inputs can be solved in $2^{n - \log^k n}$ time for any $k$. 
	
	By Proposition~\ref{prop:circuit-facts-contain} (9), a size $s = 2^{\log^k n}$ size $\AND_3 \circ \SYMAND$ circuit can be transformed into a $\poly(s)$ size $\SYMAND$ circuit, which can in turn be transformed to a $\poly(s) = 2^{O(\log^k n)}$ size $\OR \circ \MAJ \circ \OR$ circuit by Lemma~\ref{lm:reduction}. Note that by our hypothesis, we have an algorithm solving $\MAXSAT$ with $2^{\log^{k'} n}$  clauses in $2^{n - \log^{k'} n}$ time for any $k'$, and this algorithm can be used to solve the satisfiability for $\MAJ \circ \OR$ sub-circuits. 
	
	Therefore, the satisfiability of a $2^{\log^k n}$ size $\AND_3 \circ \SYMAND$ circuit can be solved in $2^{n + O(\log n)} / 2^{\log^{k'} n } \cdot 2^{O(\log^k n)} \le 2^{n - \log^{(k'-1)} n}$ time for large enough $k'$. Then the proof is completed by applying Theorem~\ref{theo:NEXP-lowb-quasi-poly}.
\end{proof}

\begin{reminder}{Theorem~\ref{theo:MAX-SAT-ENP}}
	If there is a $2^{n - \Omega(n/\log m)}$ time algorithm for $\MAXSAT$ with $m$ clauses. Then $\ENP$ does not have $2^{o(\sqrt{n})}$-size $\SYMAND$ circuits.
\end{reminder}
\begin{proof}
	By Theorem~\ref{theo:ENP-lowb}, it suffices to show that the satisfiability of $s = 2^{o(\sqrt{n})}$ size $\AND_3 \circ \SYMAND$ circuits with $n+O(\log n)$ inputs can be solved in $2^{n} /n^{\omega(1)}$ time.
	
	Again, with the same step in the proof of Theorem~\ref{theo:MAX-SAT-NEXP}, this size $s$ $\AND_3 \circ \SYMAND$ circuit can be transformed into an equivalent $\poly(s) = 2^{o(\sqrt{n})}$ size $\OR \circ \MAJ \circ \OR$ circuit. With our $\MAXSAT$ algorithm, the satisfiability of the latter circuit can be decided in
	\[
	2^{n + O(\log n) - n/o(\sqrt{n}) + o(\sqrt{n})} = 2^{n - \omega(\sqrt{n})}
	\]
	time, which completes the proof.
\end{proof}

\section{$\KSAT$}



We need the following Lemma from~\cite{abboud2017more}.

\begin{lemma}[Lemma 4.8 in~\cite{abboud2017more}]\label{lm:TC-to-KSAT}
	There is a polynomial-time many-one reduction from \textsf{TC-SAT} to \textsf{CNF-SAT} that, given $\eps \in (0,1)$ and a depth-$d$ threshold circuit with at most $cn$ wires, with $c \ge 1$, produces a $k\textsf{-CNF}$ formula $\varphi$ on at most $(1+\eps) n$ variables and with
	\[
	k \le (2000 (c/\eps) \log(2c/\eps))^{d} + 1.
	\] 
\end{lemma}

\begin{theo}
	A $2^{n \cdot (1 - 1/k^{1 / \omega(\log\log k)} )}$ time algorithm for $\KSAT$ implies that for any constant $c > 1$, $\ENP$ has no $cn$-wire depth-$(c\log\log n)$ $\TC$ circuit.	
\end{theo}
\begin{proof}
	For any constant $c$, suppose we are given a circuit of $(cn/3 - 1)$-wire and depth-$(c \log \log n - 1)$, in order to apply Theorem~\ref{theo:ENP-lowb}, we need to show the $\AND$ of $3$ such circuits admits a faster satisfiability algorithm. 
	
	Note that $\AND$ of $3$ such circuits is just a $\TC$ circuit of $cn$-wire and depth-$(c \log \log n)$, denote that circuit by $C$. Let $\eps$ be a parameter to be decided later, we apply Lemma~\ref{lm:TC-to-KSAT} to transform the satisfiability problem of $C$ into a $k\textsf{-CNF}$ formula $\varphi$ on $(1+\eps) n$ variables, with
	\[
	k \le (2000 (c/\eps) \lg(4c/\eps))^{ c\log\log n} + 1.
	\]
	
	Now we set $\eps$ so that $\eps^{-1} = 2^{\log^{t} n}$ for a small constant $t$. We then have $\log \log k = \Theta(\log \log n)$.
	 
	Applying the assumed $k$-$\SAT$ algorithm, the running time can be calculated as
	\begin{align*}
	2^{(1+\eps) (1 - 1/k^{1 / \omega(\log\log k)} ) n } = 2^{(1+\eps) \cdot n \cdot (1 - (\eps/c)^{o(1)})} = 2^{n (1 - \eps^{o(1)})} = 2^{n} / \log^{\omega(1)} n.
	\end{align*}
	
	The proof is completed by applying Theorem~\ref{theo:ENP-lowb}.
\end{proof}




\section*{Acknowledgment}
I would like to thank Ryan Williams for detailed comments on an early draft of this paper, countless helpful discussions and encouragements during this work, and pointing out some applications of the structure lemmas for $\THRTHR$. 

I am grateful to Ofer Grossman, Kaifeng Lv and Peilin Zhong for helpful discussions and suggestions.
	
\bibliographystyle{alpha}
\bibliography{team}

\appendix
\section{An Alternative Proof for Lemma~\ref{lm:THRTHR-to-WMaxIP}}

Here we present an alternative proof for Lemma~\ref{lm:THRTHR-to-WMaxIP}, which reduces the satisfiability problem for $\THRTHR$ to $\WMaxIP$ directly, without applying Corollary~\ref{cor:equivalent-THRTHR-THRMAJ}.

\begin{proofof}{Lemma~\ref{lm:THRTHR-to-WMaxIP}}
	We are going to apply Lemma~\ref{lm:lowb-THRTHR} by reducing the satisfiability problem for $\THRTHR$ circuits to $\WMaxIP$ or $\Hopcroft$.
	
	Our reduction here roughly follows Theorem~3.1 of~\cite{Williams14THR}, which is itself inspired by~\cite{Matousek91Dom}.
	
	Let $\LEQ : \mathbb{Z} \times \mathbb{Z} \to \{0,1\}$ be the function that $\LEQ(a,b) := 1$ if $a \le b$ and $0$ otherwise. For simplicity, we assume $n$ is even. Let $C$ be a $\THRTHR$ circuit of size $s = n^k$ and $G$ be its top $\THR$ gate. Let $W_1,W_2,\dotsc,W_{s}$, $T$ and $L$ be the weights, threshold and associate linear function of $G$. Let $G_1,G_2,\dotsc,G_{s}$ be the corresponding $\THR$ gates on the bottom layers. We use $L_1,L_2,\dotsc,L_{s}$ and $T_{1},T_{2},\dotsc,T_{s}$ to denote their associated linear functions and thresholds.
	
	For each $x,y \in \{0,1\}^{n/2}$, we interpret $x$ and $y$ as an assignment to the first half and second half of the input to $C$ respectively.
	
	For each linear functions $L_{j}$, we use $X_j(x)$ and $Y_j(y)$ to denote the contribution from $x$ and $y$ respectively. We have
	\[
	G_{j}(x,y) := \LEQ(T_j,L_{j}(x,y)) = \LEQ(T_j,X_j(x) + Y_j(y)) = \LEQ(T_j - X_j(x), Y_j(y)).
	\]
	
	And therefore
	\[
	L(x,y) = \sum_{j=1}^{s} W_j \cdot \LEQ(T_j - X_j(x), Y_j(y)).
	\]
	
	Then, for each $x \in \{0,1\}^{n/2}$, we construct the vector $A(x)$, such that $A(x)_j := T_j - X_j(x)$. Similarly, for each $y \in \{0,1\}^{n/2}$, we construct vector $B(y)$ with $B(y)_j := Y_j(y)$.
	
	Now, let $N = 2^{n/2}$. For each $j \in [s]$, let $S_j$ be the sorted list of all integers $A(x)_j$'s and $B(y)_j$'s for $x,y \in \{0,1\}^{n/2}$. If two values are the same, items from $A(x)_j$'s come first. Then we replace each $A(x)_j$'s and $B(y)_j$'s by their ranks in the list $S_j$. It is easy to see that this step reduces the weight to $[2N]$, and preserves the value of $\LEQ(A(x)_j,B(y)_j)$.
	
	Let $t$ be a parameter to be specified later, for each $j$, we partition $S_j$ into $t$ buckets, each of size at most $\lceil 2N / t \rceil$. Let $x,y \in \{0,1\}^{n/2}$ be assignments to $A$ and $B$, there are two cases:
	
	\paragraph*{There is a $j \in [s]$ such that $A(x)_j$ and $B(y)_j$ are in the same buckets.} In this case, note that for each $x \in \{0,1\}^{n/2}$, there are at most $s \cdot (2N / t)$ possibly $y$ such that $(x,y)$ belongs to this case. Hence, we can enumerate all such pairs and check them in $N^2 / t \cdot s^{c}$ time for a universal constant $c$.
	
	\paragraph*{For all $j \in [s]$, $A(x)_j$ and $B(y)_j$ are in different buckets.} In this case, we can safely replace each $A(x)_j$ and $B(y)_j$ by the indexes of their buckets, which reduces their range to $[t]$.
	
	Now we define some auxiliary vectors to ease our construction. For $k \in [t]$, we define $e^{[k]} \in \{0,1\}^{t}$ such that $e^{[k]}_i = 1$ if and only if $i = k$, we also define $o^{[k]} \in \{0,1\}^{t}$ such that $o^{[k]}_i = 1$ if and only if $k > i$. 
	
	Recall that for two vectors $u,v \in \{0,1\}^*$, we use $u \circ v$ to denote their concatenation. We define:
	\[
	u(x) := ( \circ_{j=1}^{s} e^{[A(x)_j]} ) \circ ( \circ_{j=1}^{s} e^{[A(x)_j]} ),
	\]
	\[
	v(y) := ( \circ_{j=1}^{s} o^{[B(y)_j]} ) \circ ( \circ_{j=1}^{s} e^{[B(y)_j]} ),
	\]
	\[
	w := (\circ_{j=1}^{s} (W_j)^{\otimes t} ) \circ (-M)^{\otimes (s \cdot t)}.
	\]
	
	In which $M$ denote a sufficient large number (can be treated as infinity) and $(W_j)^{\otimes t}$ denotes a vector repeating $W_j$ $t$ times.
	
	Now, consider $u(x) \odot_w v(y)$, it is straightforward to verify that that value would be very small if there exists a $j \in [s]$ such that $A(x)_j = B(y)_j$, and is equal to $L(x,y)$ otherwise.
	
	Therefore, computing the maximum of $u(x) \odot_w v(y)$ for all $x,y \in \{0,1\}^{n/2}$ solves this case, which can be reduced to a $\WMaxIP_{N,2st}$ instance.
	
	Setting $t = s^{c} \cdot n^{k'}$ for an integer $k'$. The running time becomes $O(N^2 / n^{k'})$ plus the running time for solving $\WMaxIP_{2^{n/2},O(n^{O(k) + k'})}$, which is also $O(N^{2} / n^{k'}) = O(2^{n} / n^{k'})$ by our assumption. Applying Lemma~\ref{lm:lowb-THRTHR} completes the proof.
	
	
	The reduction to $\Hopcroft$ works roughly the same, with the only modification that we transform the $\THRTHR$ circuit into an equivalent $\DISJOR\circ\ETHR\circ\THR$ at the beginning (via Proposition~\ref{prop:circuit-facts-contain} (3)), and solve each $\ETHR \circ \THR$ sub-circuits separately via a similar reduction to $\Hopcroft$.
\end{proofof}
\section{Applications of Structure Lemma I in Communication Complexity}\label{app:cc}

In this section we prove Theorem~\ref{theo:new-protocols-THRTHR}. First we introduce the formal definition of $\RP \cdot \UPPcc$ protocols.

\begin{defi}[$\RP \cdot \UPPcc$ Protocols]	
	For a problem $\Pi$ with inputs $x,y$ of length $n$ (Alice holds $x$ and Bob holds $y$), we say a communication protocol is a $\RP \cdot \UPP$ communication protocol with cost $c$ if the following holds.
	
	\begin{itemize}
		\item Alice and Bob jointly toss $c$ public coins to get a string $z \in \{0,1\}^c$.
		
		\item Given $y$ and $z$, Bob sends Alice $c$ bits, and Alice decides to accept or not.\footnote{In $\UPP$, one-way communication is equivalent to the seemingly more powerful one in which they communicate~\cite{paturi1986probabilistic}.} They have an unlimited supply of private random coins (not public, which is important) during their conversation. 
		
		\begin{itemize}
			\item If $\Pi(x,y) = 1$, for at least half $z$'s from $\{0,1\}^c$, Alice accepts with probability $> 1/2$.
			\item Otherwise, for all $z$ from $\{0,1\}^c$, Alice accepts with probability $< 1/2$.
		\end{itemize}
		
	\end{itemize}
\end{defi}

Also, we need the following standard fact about $\THR \circ \MAJ$ circuits.

\begin{lemma}[\cite{ForsterKLMSS01}]\label{lm:THR-MAJ-UPP}
	For a function $F : \{0,1\}^{n} \times \{0,1\}^n \to \{0,1\}$, suppose it admits a $\THR\circ\MAJ$ circuit of size $s$, then it also admits a $\UPP^\cc$ protocol of cost $O(\log s)$.
\end{lemma}

Then Theorem~\ref{theo:new-protocols-THRTHR} follows directly from Lemma~\ref{lm:structure-THRTHR-I} and Lemma~\ref{lm:THR-MAJ-UPP}.

\begin{proofof}{Theorem~\ref{theo:new-protocols-THRTHR}}
	Given a $\THRTHR$ circuit $C$ of size $s$ computing $F$, by Lemma~\ref{lm:structure-THRTHR-I}, it has an equivalent $\GAPOR \circ \THR \circ \MAJ$ circuit $C'$ of size $\poly(s)$. Alice and Bob first toss $O(\log s)$ public coins to select a $\THR \circ \MAJ$ sub-circuit of $C'$ at uniformly random, and simulate the $O(\log s)$ cost $\UPPcc$ protocols for it by Lemma~\ref{lm:THR-MAJ-UPP}.
\end{proofof}
	
\end{document}